\documentclass[journal]{IEEEtran}
\usepackage{float}
\usepackage{hyperref}
\hypersetup{
    bookmarksnumbered=true,  
    bookmarksopen=true,      
}
\usepackage{url}
\usepackage[section]{placeins}
\usepackage{morefloats}
\usepackage{graphicx} 
\graphicspath{ {./fig_ssn/} {./photo_ssn/}}
\usepackage{epsfig,epstopdf}
\usepackage{caption} 
\usepackage{subcaption} 
\usepackage{algorithm,algorithmic}

\usepackage{upquote} 
\usepackage{mathtools} 
\usepackage{amssymb,bm,mathrsfs} 
\usepackage{eucal} 
\allowdisplaybreaks
\newtheorem{lemma}{Lemma}
\newtheorem{theorem}{Theorem}
\newtheorem{definition}{Definition}
\newtheorem{remark}{Remark}
\newcommand{\floor}[1]{\left\lfloor{#1}\right\rfloor}
\newcommand{\amax}[1]{\mathop{\mathrm{arg\,max}}\limits_{#1}} 
\newcommand{\amin}[1]{\mathop{\mathrm{arg\,min}}\limits_{#1}}

\newcommand{\sign}{\mbox{sign}}
\newcommand{\algref}[1]{Algorithm~\textup{\ref{#1}}}

\newcommand{\figref}[1]{Fig.~\textup{\ref{#1}}} 
\newcommand{\lemref}[1]{Lemma~\textup{\ref{#1}}}

\newcommand{\rmkref}[1]{Remark~\textup{\ref{#1}}}
\newcommand{\secref}[1]{Section~\textup{\ref{#1}}}
\newcommand{\tabref}[1]{Table~\textup{\ref{#1}}}
\newcommand{\thmref}[1]{Theorem~\textup{\ref{#1}}}
\def\hat{\widehat}
\def\tilde{\widetilde}
\def\bar{\overline}
\begin{document}

\title{A Semi-Smooth Newton Algorithm for \\
High-Dimensional Nonconvex Sparse Learning}
\author{
Yueyong~Shi,
Jian Huang,
Yuling~Jiao,
and~Qinglong~Yang%
\thanks{
The work of Y. Shi was supported in part by
the National Natural Science Foundation of China
under Grant 11801531, Grant 11701571 and Grant 41572315.
The work of Y. Jiao was supported in part by
the National Science Foundation of China
under Grant 11871474 and Grant 11501579.
The work of Q. Yang was supported in part by
the National Science Foundation of China under Grant 11671311.
\emph{(Corresponding author: Qinglong Yang.)}
}%
\thanks{Y. Shi is with the School of Economics and Management,
China University of Geosciences
and Center for Resources and Environmental Economic Research,
China University of Geosciences, Wuhan 430074, China
(e-mail: syywda@whu.edu.cn).
}%
\thanks{J. Huang is with the Department of Applied Mathematics,
The Hong Kong Polytechnic University, Hong Kong 999077, China
(e-mail: j.huang@polyu.edu.hk).
}%
\thanks{Y. Jiao and Q. Yang are with the School of Statistics and Mathematics,
Zhongnan University of Economics and Law, Wuhan 430073, China
(e-mails: yulingjiaomath@whu.edu.cn; yangqinglong@zuel.edu.cn).} 
}

\markboth{IEEE TRANSACTIONS ON NEURAL NETWORKS AND LEARNING SYSTEMS,
~Vol.~XX, No.~X, August~20XX}%
{Shell \MakeLowercase{\textit{et al.}}: Bare Demo of IEEEtran.cls for IEEE Journals}

\maketitle

\begin{abstract}
The smoothly clipped absolute deviation (SCAD) and  the
minimax concave penalty (MCP) penalized regression models are
two important and widely used nonconvex  sparse learning tools that can handle
variable selection and parameter estimation simultaneously,
and thus have potential applications in various fields such as
mining biological data in high-throughput biomedical studies.
Theoretically, these  two models enjoy the oracle
property even in the high-dimensional settings, where  the number of predictors $p$ may be much larger than the number of  observations  $n$. However,  numerically,
it is quite challenging to  develop fast and stable algorithms due to  their  non-convexity and  non-smoothness. In this paper we develop a  fast   algorithm  for
SCAD and MCP penalized learning  problems.
First, we show that the global minimizers of both models are roots of the
nonsmooth equations.  Then,
a semi-smooth Newton (SSN) algorithm is employed to solve the equations.
We prove that the SSN algorithm converges  locally and superlinearly  to
the Karush-Kuhn-Tucker (KKT) points.
Computational complexity analysis shows  that the cost of the SSN algorithm per iteration is $O(np)$. Combined with the warm-start technique, the SSN algorithm
can be very efficient and accurate.
Simulation studies and a real data example suggest that
our SSN algorithm, with comparable solution accuracy with
the coordinate descent (CD) and
the difference of convex (DC) proximal Newton algorithms,
is more computationally efficient.
\end{abstract}


\IEEEpeerreviewmaketitle

\section{Introduction}
\label{sec:intro}

\IEEEPARstart{T}{his} paper introduces a fast algorithm
for concavely penalized regression. We focus on the linear
regression model
\begin{equation}\label{model}
 y = X\beta^{\dag} + \varepsilon,
\end{equation}
where $y\in \mathbb{R}^{n}$ is an $n\times 1$ vector of response variables,
$X=(X_1, \ldots, X_p)$ is an $n\times p$
design matrix, $\varepsilon$ is an $n\times 1$ vector of error terms, and
$\beta^{\dag} =(\beta^{\dag}_{1},\ldots ,
\beta^{\dag}_{p})^T \in \mathbb{R}^{p} $ is the vector of underlying
regression coefficients.

Under the sparsity assumption that the number of important predictors is relatively small,
it is natural to consider the estimator that solves the minimization problem
\begin{equation}\label{l0}
 \min_{\beta\in \mathbb{R}^{p}}\|X \beta
 -y\|^{2}_{2}   \quad \textrm{subject to}  \quad \|\beta\|_{0}
 \leq\tau,
\end{equation}
where  $\|\beta\|_0$ denotes the number of nonzero elements of $\beta$ and
$\tau>0$ is a tuning parameter controlling the sparsity level.
However,  the minimization problem
(\ref{l0}) is  NP-hard \cite{NatarajanB1995}, hence it is quite
challenging to design a feasible algorithm for solving it when $p$ is large. Replacing
the $\|\beta\|_{0}$ term in (\ref{l0}) by $\|\beta\|_{1}$,
we get the $\ell_1$ penalized problem or the LASSO \cite{TibshiraniR1996}
\begin{equation}\label{lasso}
 \min_{\beta\in \mathbb{R}^{p}}\|X \beta
 -y\|^{2}_{2} \quad \textrm{subject to}  \quad \|\beta\|_{1} \leq\tau,
\end{equation}
which can be viewed as a convex relaxation of (\ref{l0}).
Numerically, it is convenient to consider the Lagrange form of (\ref{lasso})
\begin{equation}\label{reglasso}
 \min_{\beta \in \mathbb{R}^{p}}  \frac{1}{2}\|X\beta
 -y\|^{2}_{2} + \lambda\|\beta\|_{1},
\end{equation}
which is known as the basis pursuit denoising (BPDN)
in the signal processing literature \cite{ChenS2001},
where $\lambda\geq 0$ is a tuning parameter
that controls the sparsity level of solutions.
Computationally, (\ref{reglasso}) is a convex minimization problem, therefore, several fast algorithms have been proposed for computing its global minimizer,
such as Homotopy or LARS \cite{OsborneM2000,EfronB2004}
and CD algorithms \cite{FuW1998,FriedmanJ2007,WuT2008}.

Theoretically, under certain regularity conditions on the design matrix $X$, such as  the restricted isometry property \cite{CandesE2005}, the strong
irrepresentable condition \cite{MeinshausenN2006,ZhaoP2006} and
the sparsity condition on the regression coefficients, LASSO has attractive
estimation and selection properties.
However, even under these conditions,
the minimizer of (\ref{reglasso}) still suffers from the so-called
LASSO bias, which implies that the
LASSO regularized estimator does not have  the oracle
property.  To remedy
this problem,  \cite{FanJ2001}  proposed using
concave penalties that can reduce bias and still yield sparse solutions.
This leads to the following minimization problem
\begin{equation}\label{regconcave}
 \min_{\beta \in \mathbb{R}^{p}}  \frac{1}{2}\|X\beta
 -y\|^{2}_{2} + \sum_{i=1}^p P(\beta_{i};\lambda, \gamma),
\end{equation}
where $P(\cdot; \lambda, \gamma)$ is a concave penalty function.
Here $\lambda \ge 0$ is the penalty parameter and $\gamma$ is a given parameter that controls the concavity of the penalty.
In this paper,
we focus on concave penalties SCAD \cite{FanJ2001} and MCP \cite{ZhangC2010}.

The SCAD penalty is defined as
\begin{equation}
\label{regscad}
P_{scad}(t;\lambda,\gamma )=\lambda \int_{0}^{t}\min \{1,(\gamma -x/\lambda
)_{+}/(\gamma -1)\}dx, \gamma > 2,
\end{equation}
and the MCP takes the form
\begin{equation}
\label{regmcp}
P_{mcp}(t;\lambda,\gamma )=\lambda \int_{0}^{t}(1-x/(\gamma \lambda ))_{+}dx,
\gamma > 1,
\end{equation}
where $x_{+}$ is the nonnegative part of $x$, i.e., $x_+= x 1_{\{x \ge 0\}}$.
It is noteworthy that
both penalties converge to the $\ell_1$ penalty as
$\gamma \to \infty$, and the MCP converges to the hard-thresholding penalty
as $\gamma \to 1$.
The {MCP} can be easily understood by considering its derivative,
\begin{equation}
\label{dMCPa}
\dot{P}_{mcp}(t;\lambda,\gamma) = \lambda \big(1-|t|/(\gamma \lambda)\big)_{+}\sign(t),
\end{equation}
where $\sign(t)=-1, 0,$ or $1$ if $t < 0, =0$, or $ > 0$.
The MCP provides a continuum of penalties with the $\ell_1$ penalty
at $\gamma = \infty$ and a continuous approximation of
the hard-thresholding penalty as $\gamma \to 1$.

Concavely penalized estimators
have the asymptotic oracle property under appropriate conditions
\cite{FanJ2001,ZhangC2010}.
However, it is quite challenging to solve
\eqref{regconcave} with (\ref{regscad}) and  (\ref{regmcp}) numerically,
since the objective functions to be minimized are both nonconvex and nonsmooth.
Several methods have been proposed to deal with this difficulty.
The first type of the methods can be viewed as special cases of
the MM algorithm \cite{LangeK2000}
or of multi-stage convex relaxation \cite{ZhangT2010},
such as  local quadratic approximation (LQA) \cite{FanJ2001}
and local linear approximation (LLA)  \cite{ZouH2008}.
Such algorithms generate a solution sequence $\{\beta^k\}_{k}$ that can guarantee the convergence of the objective function, but
the convergence property of
the iterated solution sequence $\{\beta^k\}_{k}$  is generally unknown.
Moreover, the cost per iteration of this type of algorithms is
the cost of a LASSO solver.
The second type of the methods include
coordinate descent (CD) type algorithms
\cite{BrehenyP2011,MazumderR2011}.
The best convergence result of  CD algorithms for minimizing \eqref{regconcave}
is that any cluster point  of $\{\beta^k\}_{k}$
must be a stationary point of (\ref{regscad}) and  (\ref{regmcp})
\cite{BrehenyP2011,MazumderR2011}.
As shown in \cite{BrehenyP2011,MazumderR2011},  CD-type algorithms are faster than the first type of algorithms mentioned above,
because their cost per iteration is only $O(np)$.
However,  CD-type algorithms may need
a large number of iterations when high accuracy is pursued,
since their convergence rates are only sublinear or
locally linear \cite{li2018calculus}.

In this paper we develop a local but superlinearly convergent algorithm for minimizing
\eqref{regconcave} with SCAD and MCP.
 The main contributions of this paper are threefold.
First, we  establish that the global minimizers of \eqref{regconcave}
with SCAD (\ref{regscad}) and  MCP (\ref{regmcp})
are roots of the nonsmooth KKT
equations. Conversely, we show that any root
of the KKT equations is at least a global coordinate-wise minimizer and stationary point of  \eqref{regconcave}.
Then we adopt the SSN algorithm
\cite{Kummer:1988,QiSun:1993,ItoK2008}
  to solve the nonsmooth KKT equations.
   Second, we establish  the  local superlinear
  convergence property of  SSN.  Furthermore,
  the computational complexity analysis shows that  the cost of each iteration
  in SSN is at most $O(np),$ which is the same as CD algorithms.
  Hence, for a given $\lambda$ and $\gamma$, the overall cost of using SSN  to find a (local) minimizer of (\ref{regconcave})  is still  $O(np)$,  since SSN always converges after only a few iterations if it is warm started. Thus SSN is possibly one of the fastest and most accurate algorithms for computing the whole solution path of (\ref{regconcave})  by running  SSN repeatedly
  at some given $\{\lambda_t\}_t$ with warm start.
  Third,  we conduct extensive numerical experiments
   to demonstrate the efficiency and accuracy of SSN,
as well as the feasibility of proposed
tuning parameter selection rules.
The comparison results with a CD and a DC Newton-type algorithm
verify the effectiveness of SSN and the tuning parameter selectors.

The remainder of this paper is organized as follows.
In \secref{sec:algorithm},
we describe the SSN algorithm.
In \secref{sec:converge},
we establish  the local superlinear convergence to KKT points of SSN
and analyze its computational complexity.
Implementation details and numerical comparisons
on simulated and real data are given in \secref{sec:example}.
We conclude in \secref{sec:conclusion}
with some comments and suggestions for future work.

\section{Semi-smooth Newton algorithm for penalized regression}
\label{sec:algorithm}

\subsection{Notations and Background on Newton Derivative}

We first introduce the notations used throughout this paper and describe
the concepts and properties of
the Newton derivative  \cite{Kummer:1988,QiSun:1993,ChenX2000,ItoK2008}.

For a  column vector $\beta= (\beta_{1},\beta_{2},...,\beta_{p})^T\in \mathbb{R}^{p}$,
denote its $q$-norm by $\|\beta\|_q = (\sum_{i=1}^{p}|\beta_{i}|^q)^\frac{1}{q}$,
$q\in [1,\infty)$, and denote its $\ell_0$- and $\ell_\infty$- norm by
$\|\beta\|_0=|\{i:\beta_i\neq 0,1\leq i\leq p\}|$ and
$\|\beta\|_{\infty}=\max_{1\leq i\leq p}|\beta_i|$, respectively.
$X^\mathrm{T}$ is the transpose of the feature matrix  $X \in
\mathbb{R}^{n\times p}$, and $\|X\|$ denotes the operator norm of $X$
induced by the vector 2-norm.
The matrix $X$ is assumed to be columnwise normalized,
i.e., $\|X_i\|_2=1$ for $i=1,2,...,p$.
$\textbf{1}$  or $ \textbf{0}$ denote a
column vector or a matrix with elements all 1 or 0. Define $S
=\{1,2,...,p\}$.  For any $A\subseteq S $ with cardinality $|A|$,
denote $\beta_{A}\in \mathbb{R}^{|A|}$( or $X_{A}\in
\mathbb{R}^{|A|\times p})$ as the subvector (or submatrix) whose
entries (or columns) are listed in $A$. And $X_{AB}$  denotes
submatrix of $X$ whose rows and columns are listed in $A$ and $B$
respectively.
$\textrm{supp}(\beta)$ denotes the support of $\beta$,
and $\sign(z)$ denotes the entry-wise sign of a given vector $z$.

Let $F:\mathbb{R}^{m}\rightarrow \mathbb{R}^{l}$ be a nonlinear map.
\cite{Kummer:1988,QiSun:1993,ChenX2000,ItoK2008} generalized the classical  Newton-Raphson algorithm to find a root
  of  $F(z)=\textbf{0}$, where $F$ is not
 Fr\'{e}chet differentiable but only Newton differentiable in the following sense.

\begin{definition}
$F:\mathbb{R}^{m}\rightarrow \mathbb{R}^{l}$ is called Newton
differentiable at $x\in \mathbb{R}^{m}$ if there exists an open
neighborhood $N(x)$ and a family of mappings $D : N(x)\rightarrow
\mathbb{R}^{l\times m}$ such that
$$ \| F(x+h)-F(x)-D(x+h)h\|_{2}= o (\|h\|_{2}) \quad \text{for} \quad \|h\|_{2} \longrightarrow0.$$
The set of mappings $\{ D(z):z\in N(x)\}$ denoted  by $\nabla_{N}F(x)$
is called  the Newton derivative of $F$ at $x$.
\end{definition}

It can be easily seen that $\nabla_{N}F(x)$ coincides with the
Fr\'{e}chet derivative at $x$ if $F$ is continuously Fr\'{e}chet
differentiable.
Let $F_{i}:\mathbb{R}^{m}\rightarrow\mathbb{R}^{1}$ be Newton
differentiable at $x$ with Newton derivative $\nabla_{N}F_{i}(x),$ $i
= 1,2,...,l$,   then $F=(F_{1},F_{2}, ... ,F_{l})^T$ is also Newton
differentiable at $x$ with Newton derivative
\begin{equation}\label{nd2}
\nabla_{N}F(x)=(\nabla_{N}F_{i}(x), \nabla_{N}F_{2}(x),...,
\nabla_{N}F_{l}(x))^T.
\end{equation}
Furthermore, if both $F$ and $H$  are Newton differentiable at $x$ then any linear combination of them is also Newton differentiable at $x$, i.e., for any
$\theta,\mu \in \mathbb{R}^{1} $,
\begin{equation}\label{nd3}
 \nabla_{N}(\theta F +\mu G)(x) =   \theta \nabla_{N}F(x) + \mu  \nabla_{N}G(x).
 \end{equation}
Let $H:\mathbb{R}^{s}\rightarrow \mathbb{R}^{l}$ be Newton differentiable with Newton derivative $\nabla_{N}H$. Let $L\in \mathbb{R}^{s\times m}$ and
define $F(x)=H(Lx+z)$ for any given $z\in\mathbb{R}^{s}$. Then  it is easy to check by definition that the chain rule holds, i.e.,  $F(x)$ is Newton differentiable at $x$ with Newton derivative
\begin{equation}\label{nd4}
\nabla_{N}F(x) = \nabla_{N}H(Lx+z)L.
\end{equation}
In  Lemma  \ref{egN}, we will give two important  thresholding functions
that are  Newton differentiable but not
Fr\'{e}chet differentiable.

\subsection{Optimality Conditions and Semi-smooth Newton Algorithm}
In this subsection, we give a necessary condition for the global minimizers
of \eqref{regconcave} with the SCAD (\ref{regscad}) or the MCP (\ref{regmcp}) penalty.
Specifically, we show that the global minimizers satisfy a set of KKT equations, which are nonsmooth but are Newton differentiable.
Then we apply the semi-smooth Newton algorithm to solve these equations.

Now we derive the optimality conditions of the minimizers of \eqref{regconcave},
with the penalty function  $P(z;\lambda,\gamma)$ being  $P_{scad}(z;\lambda,\gamma)$
or  $P_{mcp}(z;\lambda,\gamma)$.

For a given  $t\in \mathbb{R}^{1}$, let
\begin{equation}
\label{TO1}
 T(t;\lambda, \gamma)  = \amin{z \in \mathbb{R}^1}
 \frac{1}{2}(z-t)^{2} + P(z;\lambda,\gamma)
\end{equation}
be the thresholding functions corresponding to $P(z;\lambda,\gamma)$,
which have closed forms
for both SCAD and MCP penalties
\cite{BrehenyP2011,MazumderR2011}.
\begin{lemma}\label{thexp}
Let $T(t;\lambda, \gamma)$ be defined in \eqref{TO1}.
Then for $P_{mcp}(z;\lambda,\gamma)$ and $P_{scad}(z;\lambda,\gamma)$,
it follows that
\begin{equation}\label{thmcp}
T_{mcp}(t;\lambda,\gamma)=
\left\{
\begin{aligned}
   &\frac{\mathcal{S}(t;\lambda)}{1- 1/\gamma}, \quad if \quad |t|\leq\gamma\lambda,\\
   &t, \qquad  \quad \quad if \quad|t|>\gamma\lambda.
\end{aligned}
\right.
\end{equation}
and
\begin{equation}\label{thscad}
T_{scad}(t;\lambda,\gamma)=
\left\{
\begin{aligned}
   &\mathcal{S}(t;\lambda), &\quad if \quad |t| \leq 2\lambda,\\
   &\frac{\mathcal{S}(t;\lambda\gamma/(\gamma-1))}{1-1/(\gamma-1)}, &\quad if \quad 2\lambda<|t| \leq \gamma\lambda,\\
   &t, &\quad if \quad |t|>\gamma\lambda.
\end{aligned}
\right.
\end{equation}
respectively, where the scalar function $\mathcal{S}(t;\lambda)= \max\{|t|-\lambda,0\}\sign(t) $ is the soft-thresholding function \cite{DonohoD1995}.
\end{lemma}

\begin{proof}
See Appendix   A.
\end{proof}

The following result derives the (nonsmooth) KKT equations for  the  global minimizers of \eqref{regconcave}. This result is the  basis of the SSN algorithm.

\begin{theorem}\label{th1}
Let $\hat{\beta}$ be a global minimizer of \eqref{regconcave}.
Then there exists $\hat{d} \in \mathbb{R}^p$  such that the following optimality conditions hold:
\begin{eqnarray}
 \hat{d} &=& \tilde{y} - G\hat{\beta} ,  \label{K1} \\
\hat{\beta} &=& \mathbb{T}(\hat{\beta} + \hat{d};\lambda, \gamma), \label{K2}
\end{eqnarray}
where $G = X^{T}X,$  $\tilde{y} = X^{T}y$,
and $\mathbb{T}(z;\lambda, \gamma)$ is the component-wise
thresholding operator of \eqref{TO1} for a given vector $z\in \mathbb{R}^p$.
Conversely, if there exists  $(\hat{\beta}, \hat{d})$
satisfying \eqref{K1} and \eqref{K2},
then $\hat{\beta}$ is a  stationary
point of \eqref{regconcave}.
\end{theorem}

\begin{proof}
See Appendix   B.
\end{proof}

Let
\begin{equation}
\label{Fdef1}
F(\beta;d)=
\begin{bmatrix}
  F_{1}(\beta;d) \\
  F_{2}(\beta;d) \\
\end{bmatrix}
: \mathbb{R}^p \times \mathbb{R}^p \to  \mathbb{R}^{2p},
\end{equation}
where
$F_{1}(\beta;d):=\beta  - \mathbb{T}(\hat{\beta} + \hat{d};\lambda, \gamma),$
and
$F_{2}(\beta;d):= G \beta + d -\tilde{y}.$
For simplicity, we refer to $F$ as a KKT function.
By Theorem \ref{th1},  the global minimizers of  \eqref{regconcave}
are roots of $F(\beta;d)$.
These roots are the stationary points of \eqref{regconcave}.
The thresholding operators corresponding to concave penalties
including SCAD and  MCP are not differentiable,
which in turn results in the non-differentiability of $F$.
This makes it difficult to find the roots of $F$.
So we resort to the SSN method
\cite{Kummer:1988,QiSun:1993,ChenX2000,ItoK2008}.

Let $z=(\beta; d)$.
At the $k${th} iteration, the SSN method
for finding the roots of $F(z)=0$ consists of two steps.

\begin{enumerate}
\item[(1)] Solve $H^{k}\delta^{k}=-F(z^{k}) $ for $\delta^{k}$, where
$H^{k}$ is an element of $\nabla_{N}F(z^{k})$.

\item[(2)] Update $z^{k+1} = z^{k} + \delta^{k}$, set $k \leftarrow k+1$
and go to step (1).
\end{enumerate}

This has the same form as the classical Newton method, except that
here we use an element of $\nabla_{N}F(z^{k})$ in step (1).
Indeed, the key to the success of this method is to find a suitable
and invertible $H^k$.
We describe the pseudocode for the SSN method in  \algref{ssn1}.

\medskip
\begin{algorithm}[H]
\caption{SSN  for finding  a root $z^*$ of $F(z)$}
\label{ssn1}
\begin{algorithmic}[1]
\STATE Input: initial guess  $ z^{0}$. Set $k=0$.
\FOR{$k=0,1,2,3,\cdots$}
\STATE Choose $H^k\in \nabla_{N}F(z^{k})$.
\STATE Get the semi-smooth Newton direction $ \delta^{k}$ by solving
\begin{equation}\label{ssnd}
H^k\delta^{k}=-F(z^{k}).
\end{equation}
\STATE Update  $z^{k+1} = z^{k} + \delta^{k}$.
\STATE Stop or $k := k+1$. Denote the last iteration by $\hat{z}$.
\ENDFOR
\STATE Output: $\hat{z}$ as a  estimation of $z^*$.
\end{algorithmic}
\end{algorithm}

\subsection{The Newton Derivatives of the KKT Functions} 
Denote the KKT functions as defined in \eqref{Fdef1}
by $F_{scad}$ and $F_{mcp}$ for SCAD and MCP, respectively.
To compute the roots of $F_{mcp}$ and $F_{scad}$ based on the SSN method,
we need to calculate their Newton derivatives.

\begin{lemma}\label{egN}
${T}_{mcp}(t;\lambda,\gamma)$ and ${T}_{scad}(t;\lambda,\gamma)$ are Newton differentiable with respect to $t$ with Newton derivatives
\begin{equation}\label{ndmcp}
 \nabla_{N} {T}_{mcp}(t) =
  \left\{
    \begin{array}{ll}
   $ 0$,    \quad &\text{$|t|<\lambda$,}\\
   $\text{$ r $ }$ \in \mathbb{R}^{1},  \quad &\text{$ |t|  = \lambda$,}\\
   $ 1/(1-1/$\gamma$)$  ,  \quad &\text{$\lambda<|t|<\gamma\lambda$,}\\
   $\text{$ r $ }$ \in \mathbb{R}^{1},  \quad &\text{$ |t|  = \gamma\lambda$,}\\
   $ 1$  ,  \quad &\text{$\gamma\lambda<|t|$.}\\
    \end{array}
  \right.
\end{equation}
and
\begin{equation}\label{ndscad}
 \nabla_{N} {T}_{scad}(t) =
  \left\{
    \begin{array}{ll}
   $ 0$,    \quad &\text{$|t|<\lambda$,}\\
   $\text{$ r $ }$ \in \mathbb{R}^{1},  \quad &\text{$ |t|  = \lambda$,}\\
   $ \text{1}$  ,  \quad &\text{$\lambda<|t|<2\lambda$,}\\
   $\text{$ r $ }$ \in \mathbb{R}^{1},  \quad &\text{$ |t|  = 2\lambda$,}\\
   $ \text{1/(1-1/($\gamma$-1))}$  ,  \quad &\text{$2\lambda<|t|<\gamma\lambda$,}\\
   $\text{$ r $ }$ \in \mathbb{R}^{1},  \quad &\text{$ |t|  = \gamma\lambda$,}\\
   $ 1 $  ,  \quad &\text{$\gamma\lambda<|t|$.}
    \end{array}
  \right.
\end{equation}
respectively.
\end{lemma}
\begin{proof}
See Appendix  C.
\end{proof}

\subsubsection{The Newton derivative of $F_{mcp}$}
Consider the KKT function $F_{mcp}$.
For any given point $z^k=(\beta^k;d^k)\in  \mathbb{R}^{2p}$, define
\begin{eqnarray}\label{aicmcp1}
A^{1}_{k} &=&
\{i \in S: \lambda<|\beta^k_{i} + d^k_{i}| < \lambda\gamma\}, \\
\label{aicmcp2}
 A^{2}_{k} & = &
 \{i \in S: |\beta^k_{i} + d^k_{i}| \geq \lambda\gamma\},\\
\label{aicmcp3}
 A_{k} &=& A^{1}_{k} \cup A^{2}_{k},\\
\label{aicmcp4}
B_{k}& =& \{i \in S: |\beta^k_{i} + d^k_{i}| \leq \lambda \}.
\end{eqnarray}
We rearrange the order of the entries of $z^k$ as follows:
$$z_{mcp}^k =(\beta^{k}_{A^{1}_{k}};d^{k}_{A^{2}_{k}};\beta^{k}_{B_{k}};d^{k}_{A^{1}_{k}};
\beta^{k}_{A^{2}_{k}};d^{k}_{B_{k}}).$$
Denote the Newton derivative of $F_{mcp}$ at $z_{mcp}^k$ as $\nabla_{N} F_{mcp}(z_{mcp}^k)$.
In Theorem \ref{th4}, we will show that $H^{k}_{mcp}\in \nabla_{N} F_{mcp}(z_{mcp}^k)$, where $H^{k}_{mcp}\in \mathbb{R}^{p\times p}$ is given by
\begin{equation}\label{aicmcp9}
H^{k}_{mcp} =
\begin{bmatrix}
H_{11}^{k}       &H_{12}^{k} \\
H_{21}^{k}       &H_{22}^{k} \\
\end{bmatrix}
\end{equation}
with
\begin{equation*}
\begin{aligned}
H_{11}^{k} &=
\begin{bmatrix}
-\frac{1}{\gamma-1}I_{A^{1}_{k}A^{1}_{k}}  & \textbf{0}              &\textbf{0}      \\
\textbf{0}                                 & -I_{A^{2}_{k}A^{2}_{k}} &\textbf{0}      \\
\textbf{0}                                 &\textbf{0}               & I_{B_{k}B_{k}} \\
\end{bmatrix}
,\\
H_{12}^{k} &=
\begin{bmatrix}
 -\frac{\gamma}{\gamma-1}I_{A^{1}_{k}A^{1}_{k}} & \textbf{0} & \textbf{0}\\
\textbf{0}                                      & \textbf{0} & \textbf{0}\\
\textbf{0}                                      & \textbf{0} & \textbf{0}\\
\end{bmatrix}
,\\
H_{21}^{k}&=
\begin{bmatrix}
G_{A^{1}_{k}A^{1}_{k}}  & \textbf{0}             &G_{A^{1}_{k}B_{k}}\\
G_{A^{2}_{k}A^{1}_{k}}  & I_{A^{2}_{k}A^{2}_{k}} &G_{A^{2}_{k}B_{k}}\\
G_{B_{k}A^{1}_{k}}      & \textbf{0}             &G_{B_{k}B_{k}}    \\
\end{bmatrix}
,\\
H_{22}^{k}&=
\begin{bmatrix}
 I_{A^{1}_{k}A^{1}_{k}} & G_{A^{1}_{k}A^{2}_{k}}  &\textbf{0}    \\
\textbf{0}              & G_{A^{2}_{k}A^{2}_{k}}  &\textbf{0}    \\
 \textbf{0}             & G_{B_{k}A^{2}_{k}}      &I_{B_{k}B_{k}}\\
\end{bmatrix}
.
\end{aligned}
\end{equation*}

\subsubsection{The Newton derivative of $F_{scad}$}
Now consider the KKT function $F_{scad}$. For any given point $z^k=(\beta^k;d^k)\in  \mathbb{R}^{2p}$, define
\begin{eqnarray}\label{aicscad1}
A^{1}_{k} &=&
\{i \in S: \lambda<|\beta^k_{i} + d^k_{i}| < 2\lambda \},\\
\label{aicscad2}
 A^{2}_{k}& =&
 \{i \in S: 2\lambda\leq|\beta^k_{i} + d^k_{i}| < \lambda\gamma\},\\
\label{aicscad3}
A^{3}_{k} &=& \{i \in S: |\beta^k_{i} + d^k_{i}|\geq\lambda\gamma \},\\
\label{aicscad4}
A^{k} &=& A^{1}_{k} \cup A^{2}_{k}\cup A^{3}_{k},\\
\label{aicscad5}
 B_{k} &=& \{i \in S: |\beta^k_{i} + d^k_{i}| \leq\lambda\}.
\end{eqnarray}
We rearrange the entries of $z_{scad}$ as follows:
$$z_{scad}^{k}=(\beta^{k}_{B_{k}};d^{k}_{A^{1}_{k}};\beta^{k}_{A^{2}_{k}};d^{k}_{A^{3}_{k}};
d^{k}_{B_{k}};\beta^{k}_{A^{1}_{k}};d^{k}_{A^{2}_{k}};\beta^{k}_{A^{3}_{k}}).$$
Denote the Newton derivative of $F_{scad}$ at $z_{scad}^k$ as $\nabla_{N} F_{scad}(z_{scad}^k)$.
In Theorem \ref{th4}, we will show that $H^{k}_{scad}\in \nabla_{N} F_{scad}(z_{scad}^k)$, where $H^{k}_{scad}\in \mathbb{R}^{p\times p}$ is given by
\begin{equation}\label{aicscad10}
H^{k}_{scad}:=
\begin{bmatrix}
  H_{11}^{k} & H_{12}^{k}\\
  H_{21}^{k} & H_{22}^{k} \\
\end{bmatrix}
\end{equation}
with
\begin{equation*}
\begin{aligned}
H_{11}^{k}&=
\begin{bmatrix}
I_{B_{k}B_{k}}& \textbf{0}              &\textbf{0}   							 &\textbf{0}                \\
\textbf{0}    & -I_{A^{1}_{k}A^{1}_{k}} & \textbf{0}    							 & \textbf{0}               \\
\textbf{0}    &\textbf{0}               &-\frac{1}{\gamma-2}I_{A^{2}_{k}A^{2}_{k}} & \textbf{0} 				 \\
\textbf{0}     & \textbf{0}              & \textbf{0}                              & -I_{A^{3}_{k}A^{3}_{k}}  \\
\end{bmatrix}
,\\
H_{12}^{k}&=
\begin{bmatrix}
\textbf{0}   & \textbf{0} & \textbf{0}   										& \textbf{0}\\
\textbf{0}   & \textbf{0} & \textbf{0}  										& \textbf{0}\\
\textbf{0}   & \textbf{0} & -\frac{\gamma-1}{\gamma-2}I_{A^{2}_{k}A^{2}_{k}}    & \textbf{0}\\
\textbf{0}   & \textbf{0} & \textbf{0}   										& \textbf{0}\\
\end{bmatrix}
,\\
H_{21}^{k}&=
\begin{bmatrix}
G_{B_{k}B_{k}}     &\textbf{0}              &G_{B_{k}A^{k}_{2}}     &\textbf{0}  			 \\
G_{A^{1}_{k}B_{k}} &I_{A^{1}_{k}A^{1}_{k}}  &G_{A^{k}_{1}A^{k}_{1}} & \textbf{0} 			 \\
G_{A^{2}_{k}B_{k}} &\textbf{0}              &G_{A^{2}_{k}A^{2}_{k}} &\textbf{0}       		 \\
G_{A^{3}_{k}B_{k}} &\textbf{0}              &G_{A^{3}_{k}A^{2}_{k}} &I_{A^{3}_{k}A^{3}_{k}}\\
\end{bmatrix}
,\\
H_{22}^{k}&=
\begin{bmatrix}
I_{B_{k}B_{k}}&G_{B_{k}A^{1}_{k}}        &\textbf{0}             &G_{B_{k}A^{3}_{k}} 		 \\
\textbf{0}    & G_{A^{1}_{k}A^{1}_{k}}   &\textbf{0}             & G_{A^{1}_{k}A^{3}_{k}} \\
\textbf{0}    & G_{A^{2}_{k}A^{1}_{k}}   &I_{A^{2}_{k}A^{2}_{k}} &G_{A^{2}_{k}A^{3}_{k}}  \\
\textbf{0}    & G_{A^{3}_{k}A^{1}_{k}}   &\textbf{0}             &G_{A^{3}_{k}A^{3}_{k}}  \\
\end{bmatrix}
.
\end{aligned}
\end{equation*}

\begin{theorem}\label{th4} Both
 $F_{mcp}$ and $F_{scad}$ are Newton differentiable at  $z_{mcp}^k$ and $z_{scad}^k$ with
$$H_{mcp}^k \in \nabla_{N}F_{mcp}(z_{mcp}^k),$$ and $$H_{scad}^k \in
\nabla_{N}F_{scad}(z_{scad}^k),$$ respectively.
Furthermore, the inverses of
$H_{mcp}^k$ and  $H_{scad}^k$  are uniformly bounded with
\begin{equation*}
\|(H_{mcp}^{k})^{-1}\|\leq M_{\gamma}
\end{equation*}
and
\begin{equation*}
\|(H_{scad}^{k})^{-1}\|\leq  M_{\gamma},
\end{equation*}
where $M_{\gamma}=(3\gamma+2)+(\gamma+1)(2\gamma+5)$.
\end{theorem}
\begin{proof}
See Appendix  D.
\end{proof}

With the Newton derivatives at hand we can apply SSN to compute the roots
of $F_{mcp}$ and $F_{scad}$. We first give the details for $F_{mcp}$.
By the definitions of $A_{k}^{1}, A_{k}^{2}, I_{k}$ and $T_{mcp}$,
we have
{\footnotesize
\begin{equation}\label{FF}
F_{mcp}(z_{mcp}^{k})=
\begin{bmatrix}
 \beta^k_{A_{k}^{1}}-\frac{\gamma}{\gamma-1}(\beta^k_{A_{k}^{1}}+d^k_{A_{k}^{1}}-
 \lambda \sign(\beta^k_{A_{k}^{1}}+d^k_{A_{k}^{1}}))\\
\beta^k_{A_{k}^{2}}-(\beta^k_{A_{k}^{2}}+d^k_{A_{k}^{2}})\\
\beta^{k}_{B_{k}}\\
G_{A^{k}_{1}A^{k}_{1}}\beta^{k}_{A^{k}_{1}}+G_{A^{k}_{1}A^{k}_{2}}\beta^k_{A^{2}_{k}} +
 G_{A^{k}_{1}B_{k}}\beta^k_{B_{k}} +  d^{k}_{A^{1}_{k}}-\tilde{y}_{A^{1}_{k}}\\
G_{A^{2}_{k}A^{1}_{k}}\beta^k_{A^{1}_{k}} + G_{A^{2}_{k}A^{2}_{k}}\beta^{k}_{A^{2}_{k}}+
 G_{A^{2}_{k}B_{k}}\beta^k_{B_{k}} + d^{k}_{A^{2}_{k}}-\tilde{y}_{A^{2}_{k}}\\
G_{B_{k}A_{k}^{1}}\beta^k_{A^{1}_{k}}
+ G_{B_{k}A_{k}^{2}}\beta^k_{A^{2}_{k}}+G_{B_{k}B_{k}}\beta^{k}_{B_{k}}+d^{k}_{B_{k}}-\tilde{y}_{B_{k}}
\end{bmatrix}.
\end{equation}
}

Substituting (\ref{FF}) and (\ref{aicmcp9}) into the
SSN direction equation
$$H^k_{mcp}\delta_{mcp}^k=-F_{mcp}(z^k_{mcp})$$ and noting that
$$z_{mcp}^{k+1} = z_{mcp}^k + \delta_{mcp}^k,$$
we get (after some tedious algebra)
\begin{align}
d^{k+1}_{A^{2}_{k}} &= \textbf{0}, \label{e211} \\
\beta^{k+1}_{B_{k}} &= \textbf{0},  \label{e212}\\
\widetilde{G}_{A_{k}A_{k}}\beta^{k+1}_{A_{k}}&= s_{A_{k}} ,  \label{e213}\\
d_{A^{1}_{k}}^{k+1} &= -\beta^{k+1}_{A_{k}^{1}}/\gamma  +
s_{A_{k}^{1}}
,\label{e214}\\
 d_{B_{k}}^{k+1} &=
\tilde{y}_{B_{k}} - G_{B_{k}A_{k}}\beta_{A_{k}}^{k+1},
\label{e215}
\end{align}
 where
\begin{eqnarray}\label{formg1}
\widetilde{G}_{A_{k}A_{k}} &=& G_{A_{k}A_{k}} -
\begin{bmatrix}
I_{A^{1}_{k}A^{1}_{k}} /\gamma & \textbf{0} \\
\textbf{0} & \textbf{0} \\
\end{bmatrix}, \\
\label{formg2}
 s_{A_{k}^{1}} &=& \lambda \sign(\beta^k_{A^{1}_{k}}+d^k_{A^{1}_{k}}), \\
\label{formg3}
s_{A_{k}} &=& \tilde{y}_{A_{k}} -
\begin{bmatrix}
    s_{A_{k}^{1}} \\
    \textbf{0} \\
\end{bmatrix}.
\end{eqnarray}

Then we summarize the above calculation in the following algorithm.

\begin{algorithm}[H]
\caption{SSN for finding  a root  of $F_{mcp}$ }
\label{algmcp}
\begin{algorithmic}[1]
\STATE Input:  $ X, y,\lambda, \gamma$,  initial guess  $(\beta^{0};d^{0})$. Set $k=0$.
\STATE Pre-compute $ \tilde{y}=X^{T}y$ and store it.
\FOR{$k= 0,1,2,3,\cdots$}
\STATE Compute $A^{1}_{k}, A^{2}_{k}, A_{k}, B_{k}$
by \eqref{aicmcp1} - \eqref{aicmcp4}.
\STATE $ \beta^{k+1}_{B_{k}} =
\textbf{0}.$
 \STATE $d^{k+1}_{A^{2}_{k}} = \textbf{0}$.
 \STATE  Compute
 $\widetilde{G}_{A_{k}A_{k}}, s_{A_{k}^{1}}, s_{A_{k}} $
 by \eqref{formg1}-\eqref{formg3}.
 \STATE $ \beta^{k+1}_{A_{k}}=
\widetilde{G}_{A_{k}A_{k}}^{-1} s_{A_{k}}.$
\STATE $ d_{A^{1}_{k}}^{k+1} =
-\beta^{k+1}_{A_{k}^{1}}/\gamma  + s_{A_{k}^{1}}$.
 \STATE $
d_{B_{k}}^{k+1} = \tilde{y}_{B_{k}} -
G_{B_{k}A_{k}}\beta_{A_{k}}^{k+1}. $
\STATE
Check Stop condition\\
If stop\\
Denote the last iteration by
$\beta_{\hat{A}}, \beta_{\hat{B}}, d_{\hat{A}}, d_{\hat{B}}.$\\
Else\\
$k:=k+1.$
\ENDFOR
\STATE Output:
$\hat{\beta}=(\beta_{\hat{A}};\beta_{\hat{B}})$,
$\hat{d}=(d_{\hat{A}}; d_{\hat{B}})$.
\end{algorithmic}
\end{algorithm}

Next, we derive the SSN algorithm for $F_{scad}$
in a similar fashion. Let
\begin{eqnarray}\label{formg4}
\widetilde{G}_{A_{k}A_{k}}&=&{G}_{A_{k}A_{k}}-
                                                \begin{bmatrix}
                                                  \textbf{0} & \textbf{0} & \textbf{0} \\
                                                  \textbf{0} &  I_{A^{2}_{k}A^{2}_{k}} /(\gamma-1)  & \textbf{0} \\
                                                  \textbf{0} & \textbf{0} & \textbf{0} \\
                                                \end{bmatrix}
                                              ,\\
\label{formg5}
 s_{A_{k}^{2}}&=& \frac{ \gamma\lambda}{\gamma-1}
 \sign(\beta^k_{A^{2}_{k}}+d^k_{A^{2}_{k}}),\\
 \label{formg6}
 s_{A_{k}} &=& \tilde{y}_{A_{k}} - \begin{bmatrix}
                                         d^{k+1}_{A^{1}_{k}} \\
                                          s_{A_{k}^{2}} \\
                                         \textbf{0}\\
                                       \end{bmatrix}.
\end{eqnarray}
Then the algorithm for solving $F_{scad}(z)=0$ is as follows.

\begin{algorithm}[H]
\caption{SSN for finding  a root  of
$F_{scad}$ }
\label{algscad}
\begin{algorithmic}[1]
\STATE Input:  $ X, y,\lambda, \gamma$,
initial guess  $(\beta^{0};d^{0})$. Set $k=0$.
\STATE Pre-compute $\tilde{y}=X^{T}y$ and store it.
\FOR{$k= 0,1,2,3,\cdots$}
\STATE Compute $A^{1}_{k}, A^{2}_{k},A^{3}_{k}, A_{k}, B_{k}$
by \eqref{aicscad1} - \eqref{aicscad5}.
\STATE $ \beta^{k+1}_{B_{k}} =\textbf{0}.$
\STATE $d^{k+1}_{A^{1}_{k}} = \lambda
\sign(\beta^k_{A^{1}_{k}}+d^k_{A^{1}_{k}}).$
 \STATE $d^{k+1}_{A^{3}_{k}} = \textbf{0}$.
 \STATE  Compute
 $\widetilde{G}_{A_{k}A_{k}},s_{A_{k}^{2}}, s_{A_{k}}$
 by \eqref{formg4}-\eqref{formg6}.
 \STATE $ \beta^{k+1}_{A_{k}}= \widetilde{G}_{A_{k}A_{k}}^{-1} s_{A_{k}}.$
\STATE $ d_{A^{2}_{k}}^{k+1} =
-\beta^{k+1}_{A_{k}^{2}}/(\gamma-1)  + s_{A_{k}^{2}}$.
 \STATE $
d_{B_{k}}^{k+1} = \tilde{y}_{B_{k}} -
G_{B_{k}A_{k}}\beta_{A_{k}}^{k+1}. $
\STATE
Check Stop condition\\
If stop\\
Denote the last iteration by
$\beta_{\hat{A}}, \beta_{\hat{B}}, d_{\hat{A}}, d_{\hat{B}}.$\\
Else\\
$k:=k+1.$
\ENDFOR \STATE Output:
$\hat{\beta}=(\beta_{\hat{A}};\beta_{\hat{B}})$,
$\hat{d}=(d_{\hat{A}}; d_{\hat{B}})$.
\end{algorithmic}
\end{algorithm}

A natural stopping criterion for both algorithms is when $A_{k} =
A_{k+1}$ for some $k$, which can be checked inexpensively.
We also  stop the algorithms  when the number of iterations
exceeds some given positive integer $J$ as a safeguard.

\begin{remark}\label{GammaCond}
 Here we give some remarks  on the  requirements
 for the  design matrix $X$ and  the  sparsity level of $\beta^{\dag}$.
It is obvious that  \algref{algmcp} and  \algref{algscad} are well-defined
if $\widetilde{G}_{A_{k}A_{k}}$ in  \eqref{formg1}  and
\eqref{formg4} are invertible. Sufficient conditions to
guarantee  these  are
\begin{equation}\label{speig1}
\kappa_{-}(|A_{k}|)> 1 /\gamma
\end{equation}
and
\begin{equation}\label{speig2}
\kappa_{-}(|A_{k}|) >1 /(\gamma-1),
\end{equation}
respectively,  where
$$\kappa_{-}(|A_{k}|)=\min_{\|z\|=1,\|z\|_{0}\leq|A_{k}|}\|Xz\|^2
$$ denotes the smallest  sparse eigenvalues of  $X$ with order
$|A_{k}|$ \cite{ZhangC2008}.   Therefore, we need the assumption that  the smallest singular value of the  submatrixes of $X_{A_k}$  bound away
from $1/\gamma$ (or $1/(\gamma-1)$)  to guarantee the  well-posedness  of the proposed algorithm.  This assumption is reasonable  since   the size of the active set  $A_k$ is small, due to the fact that the underlying sparsity level $\|\beta^{\dag}\|_0$ is much smaller than the sample size $n$.
\end{remark}

\section{Convergence, complexity analysis and solution path}
\label{sec:converge}

In this subsection, following \cite{Kummer:1988,QiSun:1993,ChenX2000,ItoK2008},
we establish the local superlinear  convergence  to KKT points  of
Algorithms \ref{algmcp} and \ref{algscad},   and analyze their computational complexity.
We also compare the proposed algorithm with some existing related methods.

\subsection{Convergence Analysis}
\begin{theorem}\label{th5}
Let $\beta_{mcp}^{k}$ and $\beta_{scad}^{k}$ be generated by Algorithms
\ref{algmcp} and \ref{algscad}, respectively. Then, $\beta_{mcp}^{k}$ and
$\beta_{scad}^{k}$   converge locally and superlinearly to points satisfying equation  \eqref{K1} - \eqref{K2}.
\end{theorem}
\begin{proof}
See Appendix  E.
\end{proof}

\begin{remark}
By Theorem \ref{th1} we can see that Algorithms \ref{algmcp} and \ref{algscad} also  converge locally and superlinearly to the stationary points of (\ref{regconcave}) with
MCP \eqref{regmcp} and SCAD  \eqref{regscad},  respectively.
\end{remark}

\subsection{Computational Complexity}
We now consider the number of floating point operations per iteration.
It takes $O(p)$ flops to finish steps 4-7 (steps 4-8) in
Algorithm \ref{algmcp} (Algorithm \ref{algscad}).
For step 8 (step 9)  in  Algorithm \ref{algmcp} (Algorithm \ref{algscad}), inverting  the positive definite matrix $\widetilde{G}_{A_{k}A_{k}}$ by Cholesky factorization takes $O(|A_{k}|^3)$ flops. As for steps 9-10 (steps 10-11) in  Algorithm \ref{algmcp}
(Algorithm \ref{algscad}), $O(np)$ flops are enough to finish
the matrix-vector multiplication.
 So the overall cost per iteration for Algorithms \ref{algmcp} and
 Algorithm \ref{algscad} is $O(\max(|A_{k}|^3,pn)$.
Numerically $|A_{k}|$  usually increases and converges to
$O(\|\beta^{\dag}\|_{0})$ if the algorithm is warm started.
Therefore,  if  the underlying  solution is sufficiently sparse such
that $\|\beta^{\dag}\|_{0}^3\leq O(np)$, then it takes $O(np)$ flops
per iteration for both algorithms. Hence the overall cost of
Algorithms \ref{algmcp} and \ref{algscad} are  still $O(np)$ due to their
superlinear convergence.
The computational complexity analysis
shows that  Algorithms \ref{algmcp} and \ref{algscad} are fast
to compute the solution paths with warm start.
This is supported by the numerical results presented in \secref{sec:example}.

The computational cost  of CD algorithms
\cite{BrehenyP2011,MazumderR2011}
or iterative thresholding algorithms \cite{SheY2009} for
(\ref{regscad}) and  (\ref{regmcp})  is also $O(np)$ per iteration.
The cluster points of the sequences generated by CD or iterative thresholding
also satisfy (\ref{K1})-(\ref{K2}).   But
the convergence rates of these two  types of algorithms are at best sublinear
even in the case of LASSO penalized problem where the object
function is convex \cite{BeckA2009}.
Hence it is not
surprising that Algorithms \ref{algmcp} and \ref{algscad}
outperform CD-type algorithms in efficiency
in our numerical examples given in \secref{sec:example}.

\subsection{Solution Path}
An important issue in implementing a Newton-type algorithm is how to choose
good initial values. We use warm start to determine the initial values.
This strategy has been successfully
used  for computing the LASSO and the Enet paths \cite{FriedmanJ2007}.
It is also a simple but powerful tool to handle the local convergence
of concavely penalized regression problems
\cite{BrehenyP2011,MazumderR2011,JiaoY2015}.
The local superlinear convergence results  in Theorem \ref{th5}
guarantee  that  after a small number of iterations we will get an
approximate solution with high accuracy if the algorithms  are warm
started.

We are interested in the whole solution path $\hat{\beta}(\lambda)$
for $\lambda \in [\lambda_{\min}, \lambda_{\max}]$, where $0 < \lambda_{\min}< \lambda_{\max}$ are two prespecified numbers. This will be needed for selecting
the $\lambda$ values based on a data driven procedure such as cross validation or Bayesian information criterion.
Here we approximate the solution path by computing $\hat{\beta}(\lambda)$ on a given  finite grid
$\Lambda =\{\lambda_{0},\lambda_{1}...,\lambda_{M}\}$ for some positive integer $M$,
where $\lambda_0 > \lambda_1 > \ldots > \lambda_{M} > 0$.
Obviously, $\hat{\beta}(\lambda)=\textbf{0}$ satisfies (\ref{K1}) and (\ref{K2}) if
$\lambda\geq\|X^{T}y\|_{\infty}$.
Hence we set $\lambda_{\max}=\lambda_{0}=\|X^{T} y\|_{\infty}$,
$\lambda_{t} = \lambda_{0} \rho^{t}, t =0, 1, ..., M$, and
 $\lambda_{\min}=\lambda_{0}\rho^{M}$, where $\rho\in (0,1)$.
On the decreasing sequence $\{\lambda_t\}_t$,
we use the solution at $\lambda_{t}$ as the initial value for computing the solution at $\lambda_{t+1}$, which is referred as the warm-start (or  continuation) technique.
See \cite{fan2014primal,JiaoY2015,jiao2017group,
shi2018admm,huang2018robust} and references therein for more details.
\begin{remark}\label{rmk:solpath}
In practice, we let $\lambda_{\min}=\alpha\lambda_{\max}$ for a small $\alpha$,
and then divide the interval $[\lambda_{\min},\lambda_{\max}]$ into
$M$ equally distributed subintervals in the logarithmic scale.
For a given $\alpha$, $\rho$ is determined by $M$ and a large $M$ implies a large $\rho$.
Unless otherwise specified, we set $\alpha=1e-5$ and $M=100$.
Equivalently, $\rho=\exp\{-[\log(1)-\log(\alpha)]/M\}\approx 0.9$.
Due to the local superlinear convergence property of SSN and
the warm-start technique on the solution path,
the maximum iteration number $J$ could be a small positive integer.
We recommend $J\leq 5$ and the choice $J=1$ generally works well in practice.
\end{remark}

To select an  ``optimal'' value of the tuning parameter $\lambda$,
 a high-dimensional Bayesian information criterion (HBIC)
 \cite{WangL2013,ShiY2018} can be utilized.
In this paper,  we also propose to use a novel
voting criterion (VC) \cite{huang2018robust,huang2018unified}
for choosing the optimal value of $\lambda$.
Assume we run the SSN algorithm
and obtain  a solution path until,  for example,
$\|\hat{\beta}(\lambda_t)\|_0>\floor{n/\log(p)}$ for some $t$, say $t=W$ ($W\leq M$).
Let
\begin{equation*}
\Lambda_\ell = \{\lambda_t:\|\hat{\beta}(\lambda_t)\|_0 = \ell,t=1,\ldots, W \},
\end{equation*}
where $\ell=1,\ldots,\floor{n/\log(p)}$
is the set of tuning parameters at which the
output of SSN has $\ell$ nonzero elements.
Then we determine $\lambda$ by VC as
\begin{equation}\label{vsc}
\bar{\ell}=\amax{\ell}\{|\Lambda_\ell|\}
\quad\text{and}\quad
\hat{\lambda}=\min\{\Lambda_{\bar{\ell}}\}.
\end{equation}
It is noteworthy that VC and HBIC
are  seamlessly integrated with the warm-start technique
without any extra computational overhead,  since the sequence
$\{\hat{\beta}(\lambda_t)\}$ is already generated
along the warm-start solution path.
\begin{remark}
Both VC and HBIC work well in our numerical examples.
One can use either HBIC or
VC to find solutions for simulated data,
while it is suggested to use HBIC for real-world data.
\end{remark}

\subsection{Discussion with DC Newton-type Algorithm}
\label{ssec:disDCPN}

In \cite{DCN}, the authors propose a DC proximal Newton (DCPN)
method to solve nonconvex sparse learning problems with general nonlinear loss functions.
They  rewrite a nonconvex penalty into the difference of two convex functions
and then use the multistage convex relaxation scheme
to transform the original optimizations into
sequences of LASSO regularized nonlinear regressions.
At each stage, they  propose to use the second order Taylor expansions to
approximate the nonlinear loss functions,
and then use CD with the active set strategy \cite{PCD}
to solve the deduced LASSO regularized linear regressions,
which is called proximal Newton step by the authors.
Inspired by the previous works on using nonsmooth Newton methods to find roots of
nonsmooth equations \cite{Kummer:1988,QiSun:1993,ItoK2008,Nbook},
we first derive the nonsmooth KKT equations of the optimal solutions of the target nonconvex optimization problems, and then use
the SSN method to solve the  KKT equations directly.
The proximal Newton method serves as an inner solver in  \cite{DCN},
while  the  SSN method solves the nonconvex optimization directly via finding roots of the KKT equations. Hence, it is no surprise that
SSN is generally faster than DCPN with  comparable estimation errors.
See Section \ref{ssec:comDCPN} for the numerical comparisons.

In  \cite{DCN}, the authors prove an elegant  non-asymptotic statistical estimation error bound on the solution path  (the error between the output of their algorithms at each tuning parameter and the underlying true regression coefficient) based on a restricted strong convexity assumption.
We prove the local superlinear convergence of SSN (the output of SSN converges locally and superlinearly to the KKT points).
Therefore, the work of \cite{DCN} mainly focuses on the statistical estimation error,  while we focus on the convergence rate from the optimization perspective.

\section{Numerical Examples}
\label{sec:example}

In this section we present numerical examples
to illustrate the  performance of the proposed SSN algorithm.
All the experiments are performed
in MATLAB R2010b and R version 3.3.2
on a quad-core laptop with an Intel Core i5 CPU (2.60 GHz)
and 8 GB RAM running Windows 8.1 (64 bit).

\subsection{Implementation Setting}
\label{ssec:simSet}

We generate synthetic data from \eqref{model}.
The rows of the $n\times p$ matrix $X$ are sampled as i.i.d. copies from
$N(\mathbf{0},\Sigma)$
with  $\Sigma=(r^{|j-k|})$, $1\leq j, k\leq p$,
where $r\in (0,1)$ is the correlation coefficient of $X$.
The noise vector $\varepsilon$ is generated independently from
$N(\mathbf{0},\sigma^2 I_n)$,
where $\sigma>0$ is the noise level.
The underlying regression coefficient vector $\beta^{\dag}$
is a random sparse vector chosen as $T$-sparse (i.e., $T=\|\beta^\dag\|_0$)
with a dynamic range (DR) given by
\begin{equation}\label{DR}
\text{DR}:=\frac{\max\{|\beta^\dag_i|: \beta^\dag_i\neq 0\}}
{\min\{|\beta^\dag_i|:\beta^\dag_i\neq 0\}}=10.
\end{equation}
Let $\mathcal{A}=\{i: \beta^\dag_i\neq 0\}$ be the true model and
$\hat{\mathcal{A}}=\{i:\hat{\beta}_i \neq 0\}$ be the estimated model.
Following \cite{BeckerS2011},
each nonzero entry of $\beta^{\dag}$ is generated as follows:
\begin{equation}\label{gendata}
\beta^\dag_i = \eta_{1i}10^{\eta_{2i}},
\end{equation}
where $i\in\mathcal{A}$,
$\eta_{1i}=\pm 1$ with probability $\frac{1}{2}$ and
$\eta_{2i}$ is uniformly distributed in $[0,1]$.
 Then the outcome vector $y$ is generated via $y=X\beta^\dag+\varepsilon$.
 For convenience, we use $(n,p,r,\sigma,T)$ to denote
the data generated as above.

Unless otherwise stated,
we set $\gamma = 3.7$ and $\gamma=2.7$ for SCAD and MCP, respectively,
as suggested by \cite{FanJ2001} and \cite{ZhangC2010}.

\subsection{The Convergence Behavior of the SSN Algorithm}

We give an example to
illustrate the warm-start technique and local supperlinear convergence of the SSN algorithm. To this end, we generate a simulated dataset
$(n=200,p=1000,r=0.1,\sigma=0.01,T=20)$.
To save space, we only consider SCAD for the illustration
since similar phenomenon happens
for MCP. According to the discussion in \rmkref{rmk:solpath},
we fix $J=3$ here.
Let $\widehat{\mathcal{A}}_t=\{i:\hat{\beta}_i(\lambda_t)\neq0\}$
be the active set,
where $\hat{\beta}(\lambda_t)$ is the solution to the $\lambda_t$-problem.
Set $\hat{\beta}(\lambda_0)=0$.

\begin{figure}[!ht]
\begin{subfigure}{.49\linewidth}
\centering
\includegraphics[width=\linewidth]{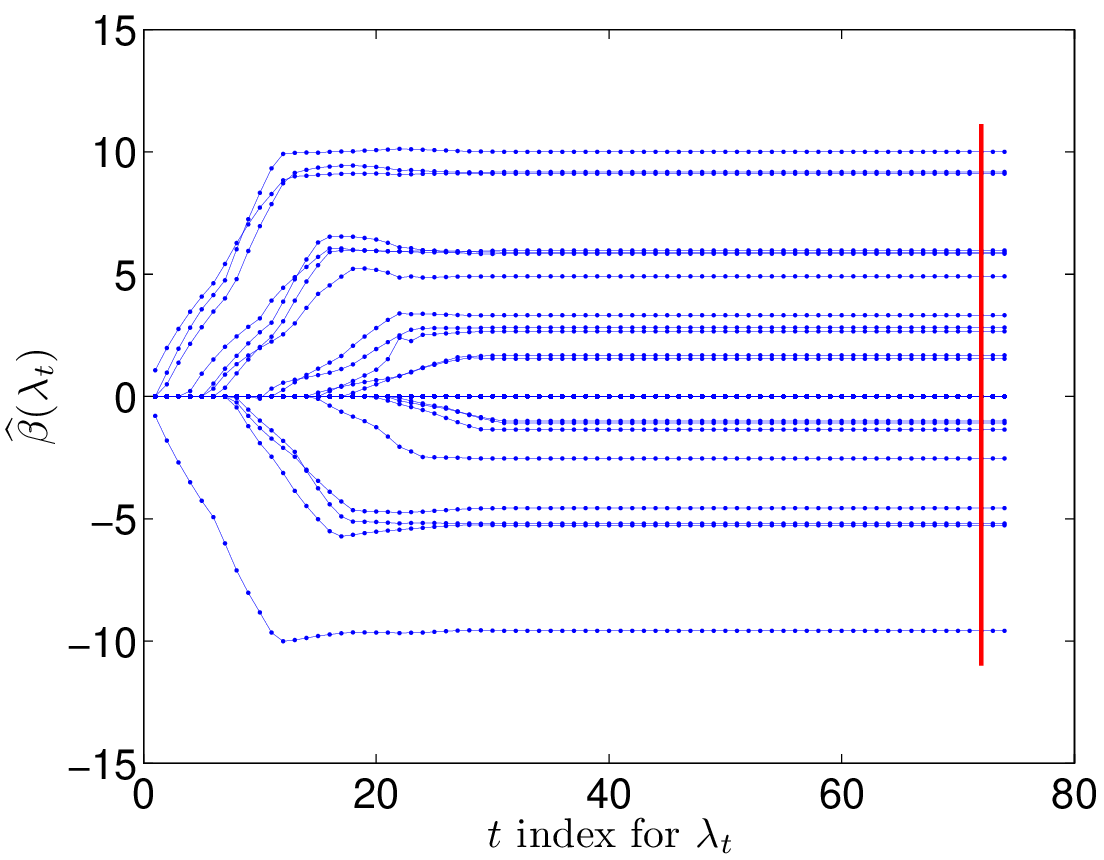}
\end{subfigure}
\begin{subfigure}{.49\linewidth}
\centering
\includegraphics[width=\linewidth]{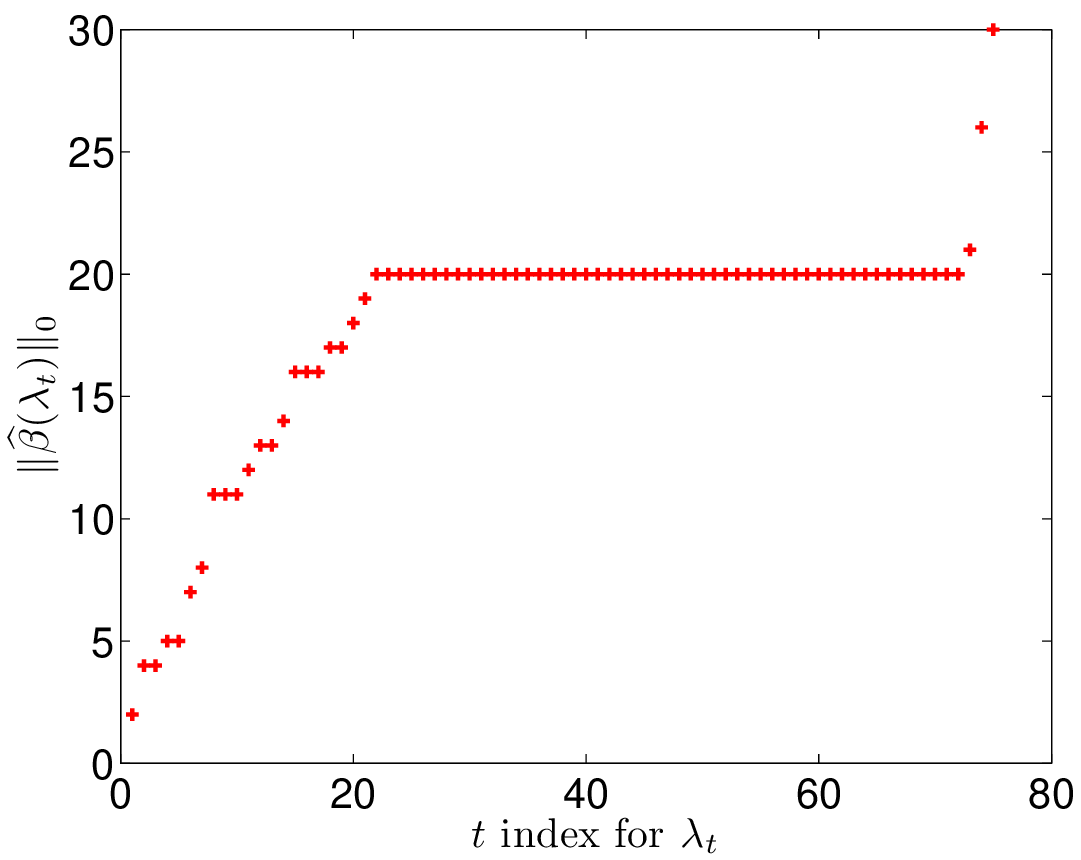}
\end{subfigure}
\caption{ The solution path of SSN  with warm start (left panel)
and the active set size along the path (right panel).
The red vertical line shows the solution selected by \eqref{vsc}.
}
\label{fig:pathact}
\end{figure}

\figref{fig:pathact} shows the solution path of SSN
and the size of the corresponding active set along the path.
It also displays $\hat{\lambda}$ and $\hat{\beta}(\hat{\lambda})$
(the red vertical line in the left panel)
determined by the VC selector \eqref{vsc}.
To examine the local convergence property of the proposed algorithm,
we present in \figref{fig:diffiter} the change of the active sets
and the number of iterations for each fixed $\lambda_t$
along the path $\lambda_0>\lambda_1>\cdots>\hat{\lambda}$.
We see that
$\widehat{\mathcal{A}}_t \subset \mathcal{A}$,
and  the size $|\widehat{\mathcal{A}}_t|$
increases monotonically as the path proceeds
and eventually equals the true model size $|\mathcal{A}|$.
In particular, at each $\lambda_{t+1}$  with $\hat{\beta}(\lambda_t)$
as the initial value,  SSN
reaches convergence within one iteration in the latter part of the path.
\figref{fig:errresest} displays
the estimation error $\|\hat{\beta}(\lambda_t)-\beta^\dag\|$ and
the prediction error (i.e., the residual) $\|X\hat{\beta}(\lambda_t)-y\|$
along the path,  as well as the underlying true $\beta^\dag$
and the solution $\hat{\beta}(\hat{\lambda})$ selected by \eqref{vsc}.

These numerical results strongly support the local superlinear convergence of the algorithm. They also demonstrate that the SSN algorithm yields estimates with good accuracy.

\begin{figure}[!ht]
\begin{subfigure}{.49\linewidth}
\centering
\includegraphics[width=\linewidth]{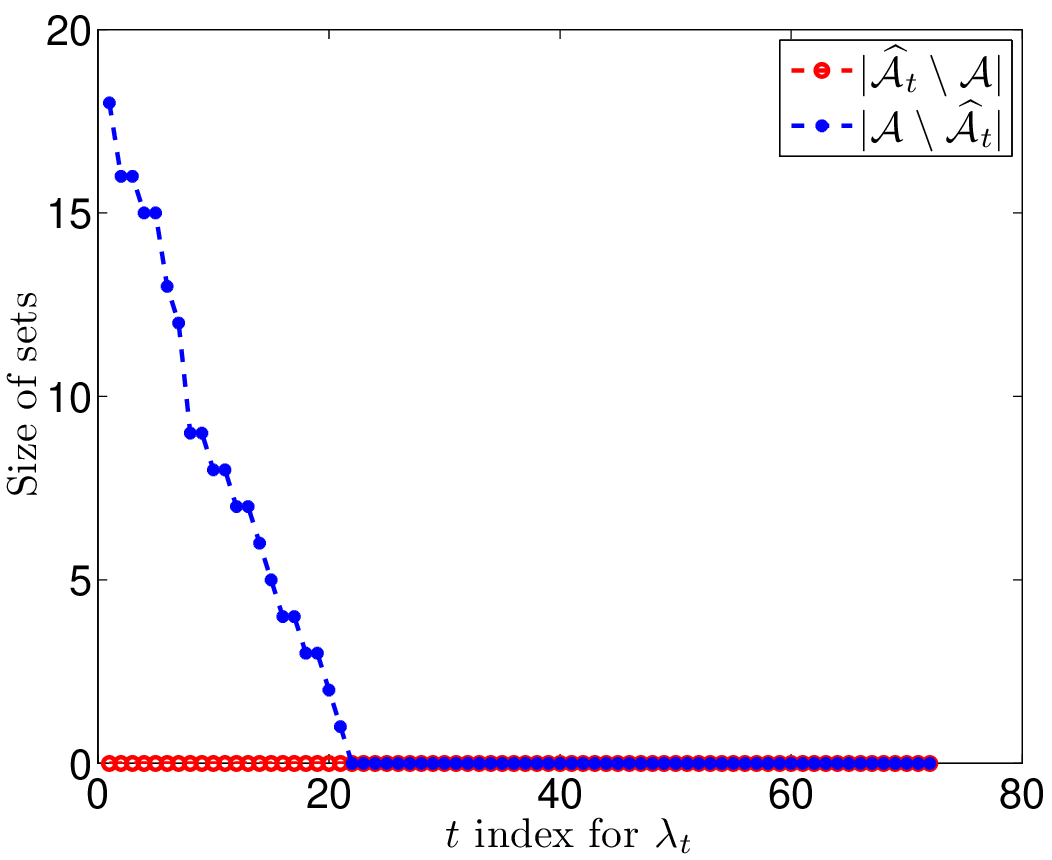}
\end{subfigure}
\begin{subfigure}{.49\linewidth}
\centering
\includegraphics[width=\linewidth]{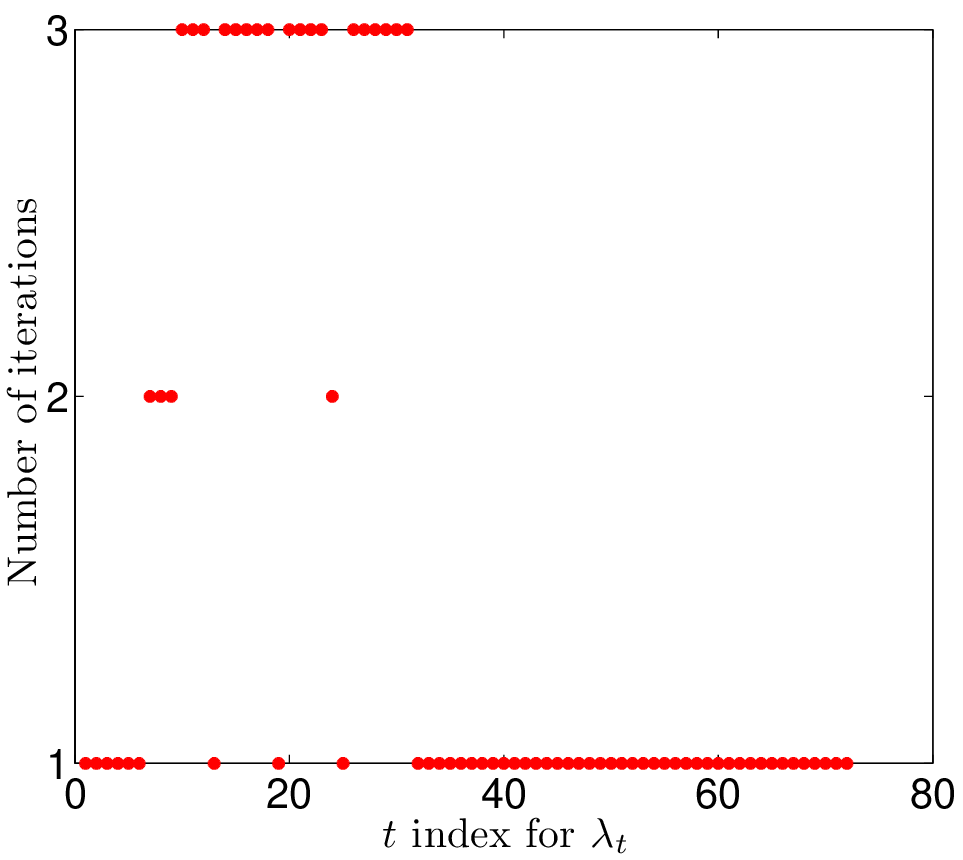}
\end{subfigure}
\caption{
The change of the active sets (left panel)
and the number of iterations (left panel)
for each $\lambda_t$-problem along the warm-start path,
where $\mathcal{A}$ is the true model
and $\widehat{\mathcal{A}}_t \backslash \mathcal{A}$
( $\mathcal{A} \backslash \widehat{\mathcal{A}}_t$)
denotes the set difference of sets $\widehat{\mathcal{A}}_t$
and $\mathcal{A}$ ($\mathcal{A}$ and $\widehat{\mathcal{A}}_t$).
}
\label{fig:diffiter}
\end{figure}

\begin{figure}[!ht]
\begin{subfigure}{.49\linewidth}
\centering
\includegraphics[width=\linewidth]{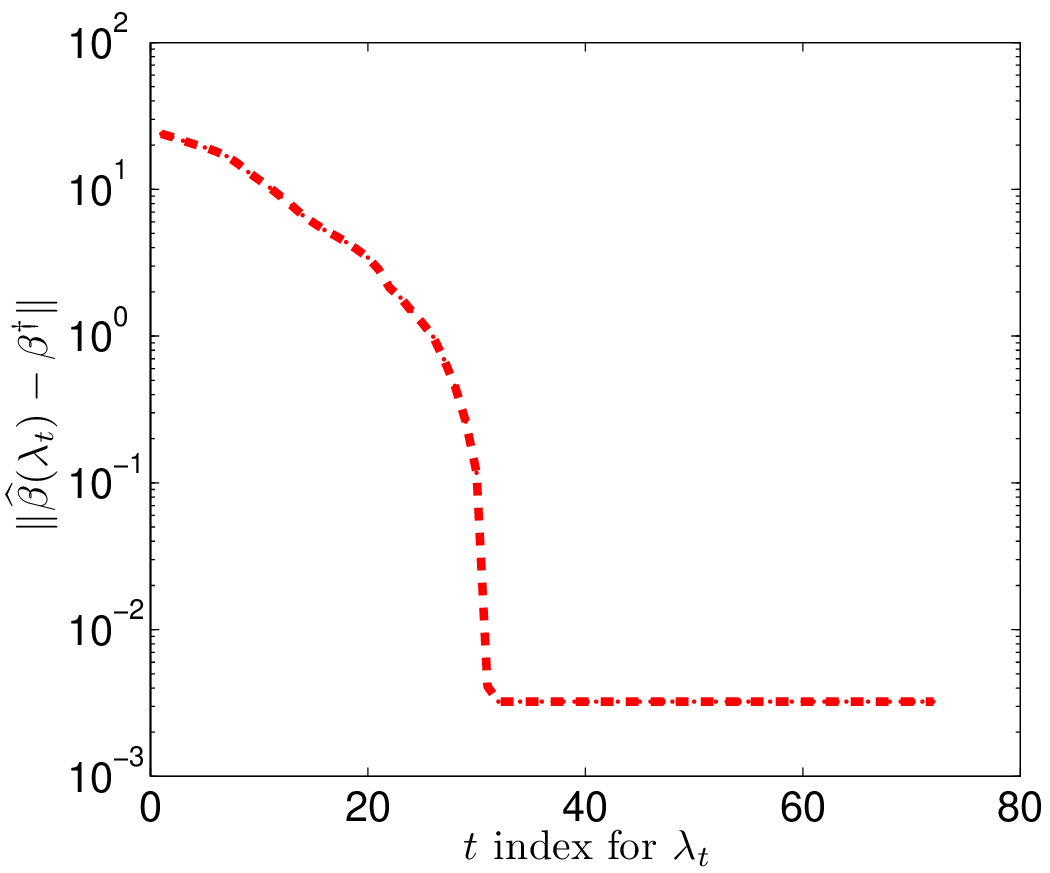}
\end{subfigure}
\begin{subfigure}{.49\linewidth}
\centering
\includegraphics[width=\linewidth]{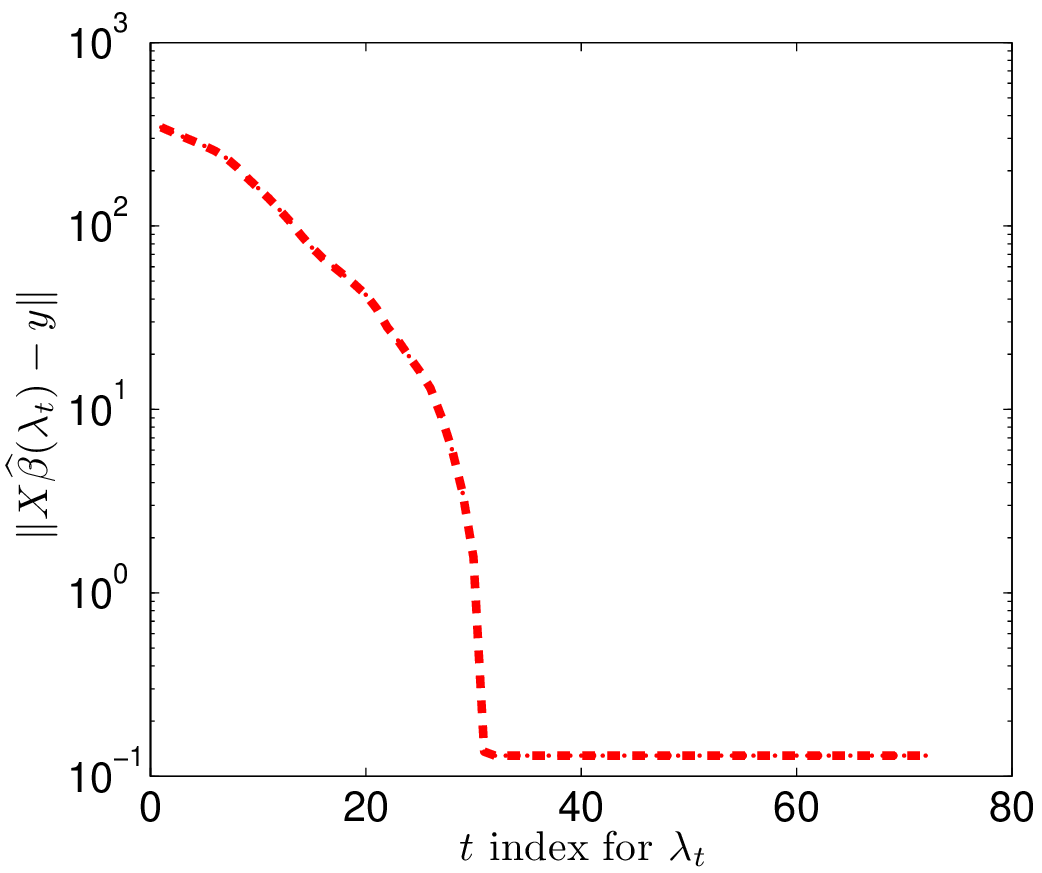}
\end{subfigure}\\[1ex]
\begin{subfigure}{\linewidth}
\centering
\includegraphics[width=0.6\linewidth]{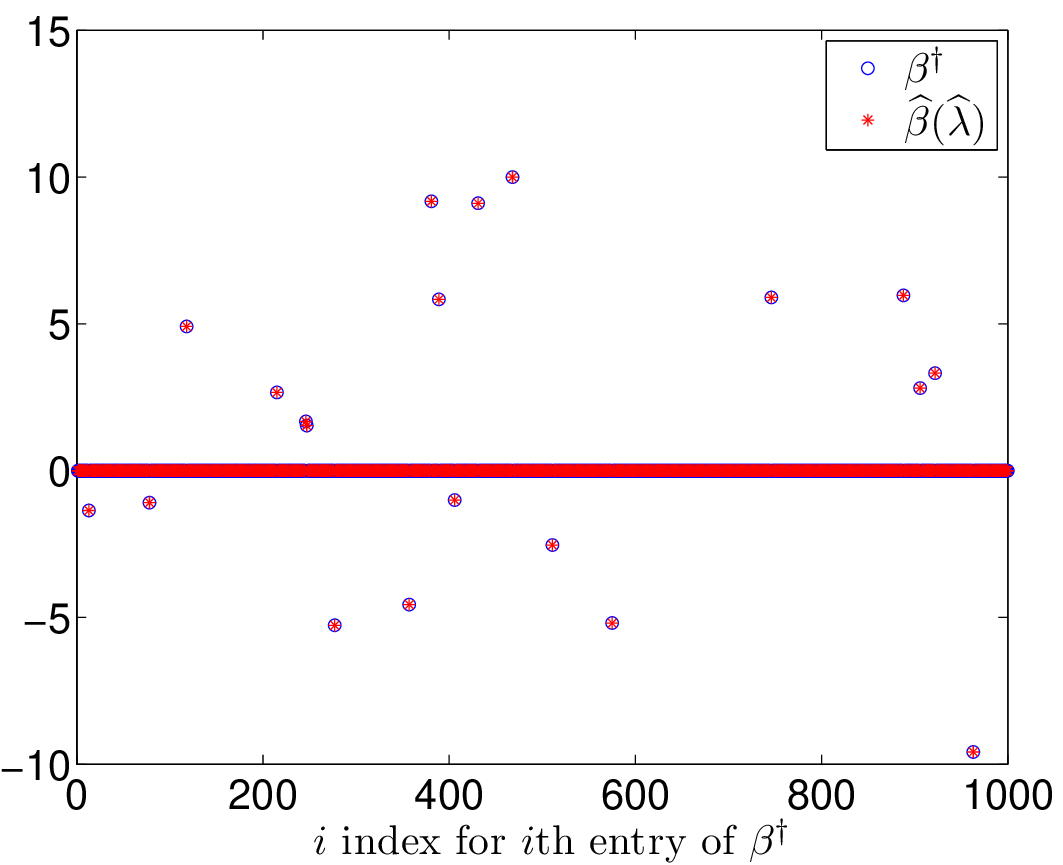}
\end{subfigure}
\caption{
The estimation error $\|\hat{\beta}(\lambda_t)-\beta^\dag\|$ (top left panel) and
the prediction error $\|X\hat{\beta}(\lambda_t)-y\|$
(top right panel) along the path, and the underlying true $\beta^\dag$
and the solution $\hat{\beta}(\hat{\lambda})$ selected by \eqref{vsc} (bottom panel).
}
\label{fig:errresest}
\end{figure}

\subsection{Comparison with A CD-type Algorithm}
\label{ssec:comCD}

We compare our SSN algorithm with a CD
algorithm in \cite{BrehenyP2011} for solving \eqref{regconcave},
which is summarized in \algref{alg:CD}.
The stopping criterion at step 8 of \algref{alg:CD}  is chosen to be either
$k\geq J$ or $\|\beta^{k+1}-\beta^k\| \leq \delta$ with $\delta=1e-3$.
We run CD on the same path as used in SSN with the
same tuning parameter selection rule.

\begin{algorithm}[!ht]
\caption{CD for MCP or SCAD}\label{alg:CD}
\begin{algorithmic}[1]
\STATE Input:
$ X, y,\lambda,\gamma$,
initial guess $\beta^{0}$ and $r^{0}=y-X\beta^0$.
Set $k=0$.
\FOR{$k= 0,1,2,3,\ldots$}
\FOR{$i=1,2,\ldots,p$}
\STATE Calculate $z_i^k=X_i^T r^k+\beta_i^k$,
where $X_i$ is the $i$th column of $X$
and $r^k=y-X\beta^k$ is the current residual value.
\STATE Update
$\beta^{k+1}_i\leftarrow T_{mcp}(z_i^k;\lambda,\gamma)$
or
$\beta^{k+1}_i\leftarrow T_{scad}(z_i^k;\lambda,\gamma)$.
\STATE Update $r^{k+1}\leftarrow r^k-(\beta^{k+1}_i-\beta^{k}_i)X_i$.
\ENDFOR
\STATE
Check Stop condition\\
If stop\\
Denote the last iteration by $\hat{\beta}$.\\
Else\\
$k:=k+1.$
\ENDFOR
\STATE Output: $\hat{\beta}$.
\end{algorithmic}
\end{algorithm}

\subsubsection{Efficiency and Accuracy}

We set $p=1000$ and $2000$ with $n=\floor{p/5}$ and $T=\floor{n/[2\log(p)]}$,
where $\floor{x}$ denotes the integer part of $x$ for $x\geq 0$.
We choose three levels of correlation $r=0.3,0.5$ and $0.7$,
which correspond to weak, moderate and strong correlation.
We consider two levels of noise: $\sigma=1$ (higher level) and $\sigma=0.1$ (lower level). The number of replications is  $N=100$.

To further illustrate the efficiency and accuracy of SSN,
we compare it with CD  in terms  of
the average CPU  time (Time, in seconds),
the average estimated model size (MS)
$N^{-1}\sum_{m=1}^N |\hat{\mathcal{A}}^{(m)}|$,
the proportion of correct models (CM, in percentage terms)
$N^{-1}\sum_{m=1}^N I\{\hat{\mathcal{A}}^{(m)}=\mathcal{A}\}$,
the average $\ell_\infty$ absolute error (AE)
$N^{-1}\sum_{m=1}^N \|\hat{\beta}^{(m)}-\beta^\dag\|_\infty$ and
the average $\ell_2$ relative error (RE)
$N^{-1}\sum_{m=1}^N (\|\hat{\beta}^{(m)}-\beta^\dag\|_2 / \|\beta^\dag\|_2)$.
The measures MS, CM, AE and RE evaluate the quality (accuracy) of the solutions.
Clearly, the closer MS approaches to $T$,  the closer CM approaches to $100\%$,
and the smaller AE and RE,  the higher the solution quality.
The results are summarized in \tabref{tab:simu}.

\begin{table}[!ht]
\centering
\caption{Simulation results for CD and SSN
with $n=\floor{p/5}$ and $T=\floor{n/[2\log(p)]}$
based on 100 independent runs. $(M,J)=(200,1)$.}
\label{tab:simu}
\scalebox{0.8}{
\begin{tabular}{cccccccccc}
\hline
 p&  r&  $\sigma$&   Penalty& Method&  Time&  MS&  CM&  AE&  RE\\
\hline
 1000&  0.3&0.1& MCP  & CD  & 1.3458&	14.00&	100\%&	0.0142&	0.0014\\
     &     &   &      & SSN & 0.1920&	14.00&	100\%&	0.0183&	0.0018\\
     &     &   & SCAD & CD  & 1.7671&	14.00&	100\%&	0.0150&	0.0015\\
     &     &   &      & SSN & 0.2068&	14.00&	100\%&	0.0188&	0.0018\\
     &     &  1& MCP  & CD  & 0.8790&	14.00&	100\%&	0.1493&	0.0147\\
     &     &   &      & SSN & 0.1228&	14.00&	100\%&	0.1504&	0.0148\\
     &     &   & SCAD & CD  & 1.1659&	13.90&	 99\%&	0.1778&	0.0174\\
     &     &   &      & SSN & 0.1418&	13.97&	 99\%&	0.1580&	0.0152\\
     &  0.5&0.1& MCP  & CD  & 1.3387&	14.00&	100\%&	0.0153&	0.0015\\
     &     &   &      & SSN & 0.1849&	13.78&	 98\%&	0.1489&	0.0132\\
     &     &   & SCAD & CD  & 1.7413&	14.00&	100\%&	0.0148&	0.0015\\
     &     &   &      & SSN & 0.2238&	14.00&	100\%&	0.0187&	0.0018\\
     &     &  1& MCP  & CD  & 0.8668&	13.92&	 98\%&	0.1730&	0.0171\\
     &     &   &      & SSN & 0.1309&	13.92&	 98\%&	0.1735&	0.0171\\
     &     &   & SCAD & CD  & 1.1434&	13.82&	 97\%&	0.1809&	0.0191\\
     &     &   &      & SSN & 0.1509&	13.97&	 99\%&	0.1492&	0.0147\\
     &  0.7&0.1& MCP  & CD  & 1.3677&	14.00&	100\%&	0.0149&	0.0015\\
     &     &   &      & SSN & 0.1773&	13.20&	 90\%&	0.3823&	0.0434\\
     &     &   & SCAD & CD  & 1.7679&	14.00&	100\%&	0.0146&	0.0015\\
     &     &   &      & SSN & 0.2168&	14.00&	100\%&	0.0190&	0.0019\\
     &     &  1& MCP  & CD  & 0.8786&	13.95&	 98\%&	0.1787&	0.0169\\
     &     &   &      & SSN & 0.1175&	12.35&	 76\%&	0.9406&	0.0983\\
     &     &   & SCAD & CD  & 1.1441&	13.47&	 91\%&	0.4013&	0.0424\\
     &     &   &      & SSN & 0.1484&	14.01&	 97\%&	0.1761&	0.0166\\
 2000&  0.3&0.1& MCP  & CD  & 3.0636&	26.00&	100\%&	0.0113&	0.0011\\
     &     &   &      & SSN & 0.7697&	26.00&	100\%&	0.0164&	0.0015\\
     &     &   & SCAD & CD  & 3.8884&	26.00&	100\%&	0.0119&	0.0011\\
     &     &   &      & SSN & 0.8128&	26.00&	100\%&	0.0169&	0.0016\\
     &     &  1& MCP  & CD  & 2.0280&	26.00&	100\%&	0.1165&	0.0108\\
     &     &   &      & SSN & 0.5568&	26.00&	100\%&	0.1169&	0.0108\\
     &     &   & SCAD & CD  & 2.5903&	26.00&	100\%&	0.1176&	0.0109\\
     &     &   &      & SSN & 0.5928&	26.00&	100\%&	0.1178&	0.0109\\
     &  0.5&0.1& MCP  & CD  & 3.0998&	26.00&	100\%&	0.0122&	0.0010\\
     &     &   &      & SSN & 0.6672&	26.00&	100\%&	0.0177&	0.0015\\
     &     &   & SCAD & CD  & 3.9453&	26.00&	100\%&	0.0117&	0.0011\\
     &     &   &      & SSN & 0.7062&	26.00&	100\%&	0.0168&	0.0015\\
     &     &  1& MCP  & CD  & 2.0312&	26.00&	100\%&	0.1204&	0.0104\\
     &     &   &      & SSN & 0.4682&	26.00&	100\%&	0.1206&	0.0105\\
     &     &   & SCAD & CD  & 2.6170&	26.00&	100\%&	0.1181&	0.0105\\
     &     &   &      & SSN & 0.5067&	26.00&	100\%&	0.1187&	0.0105\\
     &  0.7&0.1& MCP  & CD  & 3.1308&	26.00&	100\%&	0.0121&	0.0011\\
     &     &   &      & SSN & 0.6379&	24.09&	 84\%&	0.6352&	0.0607\\
     &     &   & SCAD & CD  & 3.9982&	26.00&	100\%&	0.0117&	0.0011\\
     &     &   &      & SSN & 0.7078&	26.00&	100\%&	0.0178&	0.0016\\
     &     &  1& MCP  & CD  & 2.0787&	26.00&	100\%&	0.1189&	0.0107\\
     &     &   &      & SSN & 0.5186&	24.59&	 82\%&	0.5718&	0.0561\\
     &     &   & SCAD & CD  & 2.6466&	25.93&	 97\%&	0.1547&	0.0134\\
     &     &   &      & SSN & 0.5884&	26.00&	 98\%&	0.1302&	0.0115\\	
\hline
\end{tabular}
}
\end{table}

For each $(p, r, \sigma)$ combination,
we see from \tabref{tab:simu} that SSN is faster
than CD for both MCP and SCAD.
Specifically,  SSN is about $3\sim 9$ times  faster than CD.
For a given penalty and algorithm,
the CPU time decreases as $\sigma$ increases,
increases linearly with $p$,
and is fairly robust to the choice of $r$,
with the other two parameters in the 3-tuple  $(p, r, \sigma)$ fixed.
For SCAD, SSN and CD are comparable in terms of solution quality.
Unsurprisingly, larger $\sigma$ will degrade the accuracy for both CD and SSN.
When the correlation level is low or moderate (i.e., $r=0.3$ or $0.5$),
similar phenomena hold for MCP.
For MCP with high correlation data (i.e., $r=0.7$),
CD  provides a more accurate solution than SSN does,
which is due to the choice of $\gamma=2.7$ for MCP.
In fact, simulations not reported in \tabref{tab:simu} show that
if we increase the value of $\gamma$ for MCP, say $\gamma=4$,
SSN can also produce similar solution quality comparing with CD.
Please also refer to \rmkref{GammaCond} for relevant discussions.
Overall, the simulation results in \tabref{tab:simu}  show that SSN
outperforms CD in terms of CPU time while producing solutions of comparable quality.

\subsubsection{Influence of Model Parameters}
We now consider the effects of each of
the model parameters $(\gamma,n,p,r,\sigma,T)$
on the performance of SSN and CD more closely
in terms of the exact support recovery probability (Probability),
i.e., the percentage of $\hat{\mathcal{A}}$ agrees with $\mathcal{A}$,
and the CPU time (Time, in seconds).
For the sake of simplicity,
we consider $\gamma$ for both MCP and SCAD, and only consider
$(n,p,r,\sigma,T)$ for SCAD,
since MCP has similar patterns with respect to the change in
$(n,p,r,\sigma,T)$.
Results of Probability and Time
averaged over $10$ independent runs
are given in \figref{fig:IPP} and \figref{fig:IPT}, respectively.
The parameters for the solvers are set as follows.

\paragraph{Influence of the concavity parameter $\gamma$}
Data are generated from the model with
$(\gamma\in \{1.1, 2.7, 5, 10, 20\}, n=100, p=500, r=0.1, \sigma=0.1, T=10)$
for MCP and
$(\gamma\in \{2.1, 3.7, 5, 10, 20\}, n=100, p=500, r=0.1, \sigma=0.1, T=10)$
for SCAD.

\paragraph{Influence of the sample size $n$}
Data are generated from the model with
$(\gamma=3.7, n=60:10:100, p=500, r=0.1, \sigma=0.1, T=10)$.

\paragraph{Influence of the dimension $p$}
Data are generated from the model with
$(\gamma=3.7, n=100, p=500:500:2500, r=0.1, \sigma=0.1, T=10)$.

\paragraph{Influence of the correlation level $r$}
Data are generated from the model with
$(\gamma=3.7, n=100, p=500, r=0.1:0.1:0.9, \sigma=0.1, T=10)$.

\paragraph{Influence of the noise level $\sigma$}
Data are generated from the model with
$(\gamma=3.7, n=100, p=500, r=0.1, \sigma\in\{0.2,0.4,0.8,1.6,3.2\}, T=10)$.

\paragraph{Influence of the sparsity level $T$}
Data are generated from the model with
$(\gamma=3.7, n=100, p=500, r=0.1, \sigma=0.1, T=5:5:30)$.

The results shown in \figref{fig:IPP} and \figref{fig:IPT}
demonstrate that two solvers tend to
give similar curve changing trends overall.
However, SSN  is considerably faster than CD
with better (or comparable) accuracy,
which is consistent with the simulation results in \tabref{tab:simu}.
In particular, it's good to see,  as shown in \figref{fig:IPT},
that the CPU time of SSN is insensitive to
each of the model parameters $(\gamma,n,p,r,\sigma,T)$,
which means SSN could be efficiently scaled up to other larger datasets.

\begin{figure}[!ht]
\begin{subfigure}{.49\linewidth}
\centering
\includegraphics[width=\linewidth]{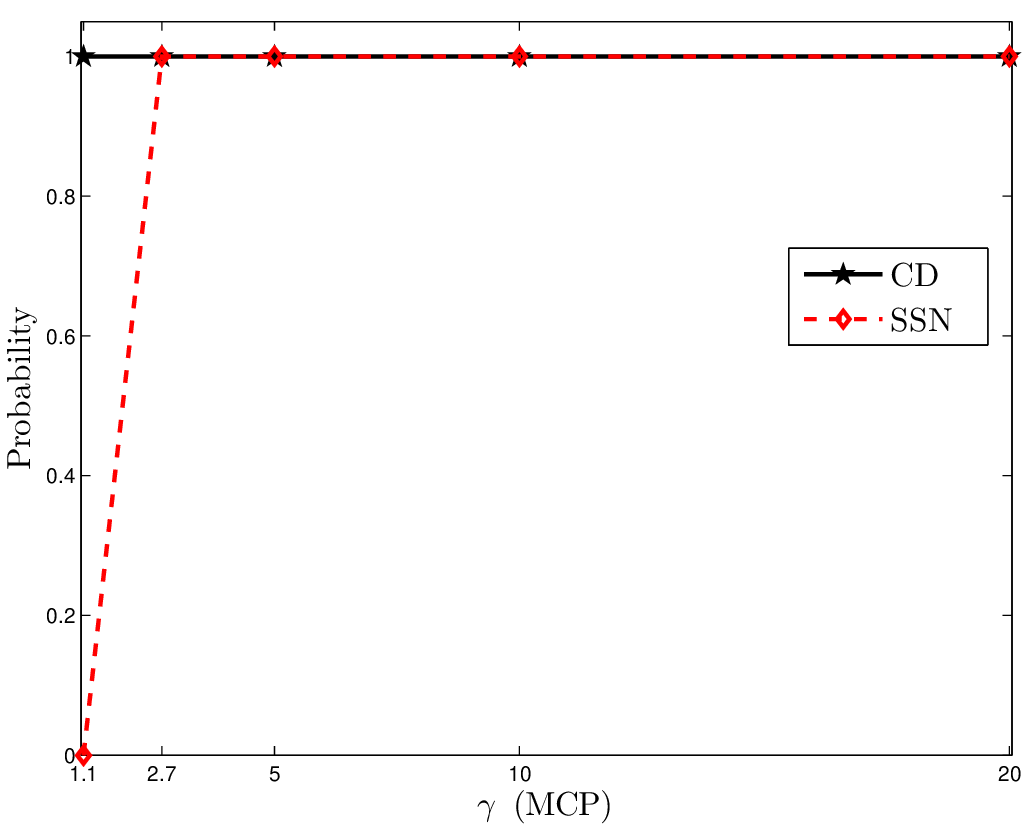}
\end{subfigure}
\begin{subfigure}{.49\linewidth}
\centering
\includegraphics[width=\linewidth]{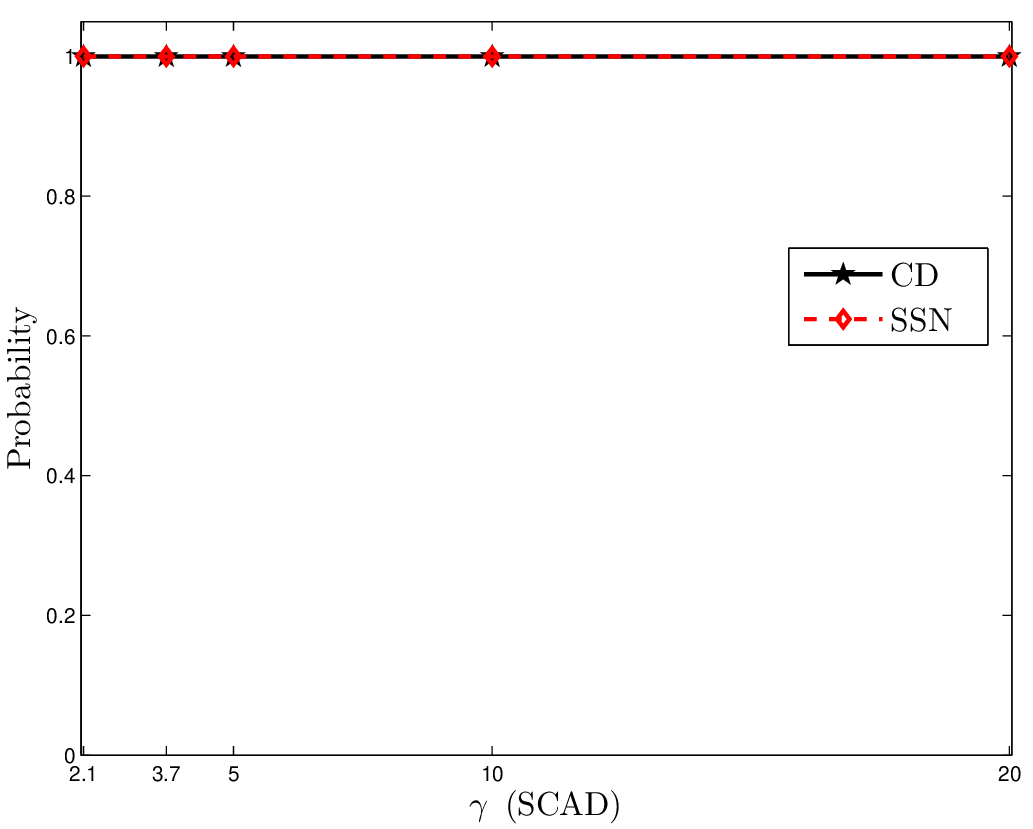}
\end{subfigure}
\begin{subfigure}{.49\linewidth}
\centering
\includegraphics[width=\linewidth]{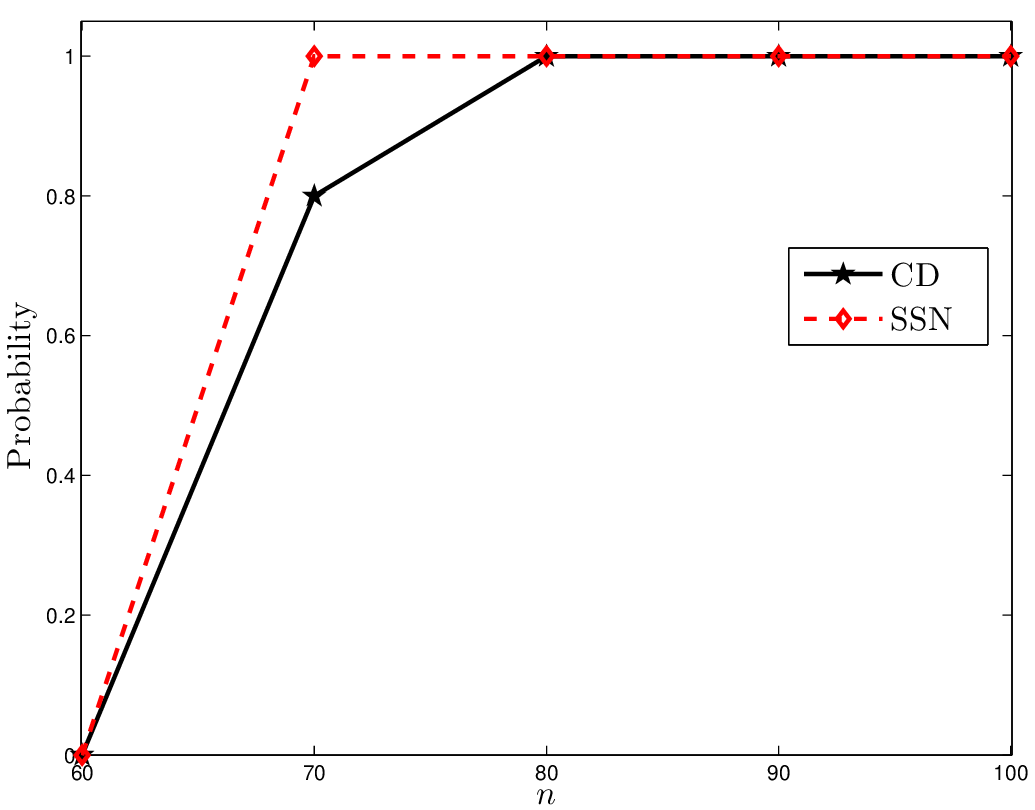}
\end{subfigure}
\begin{subfigure}{.49\linewidth}
\centering
\includegraphics[width=\linewidth]{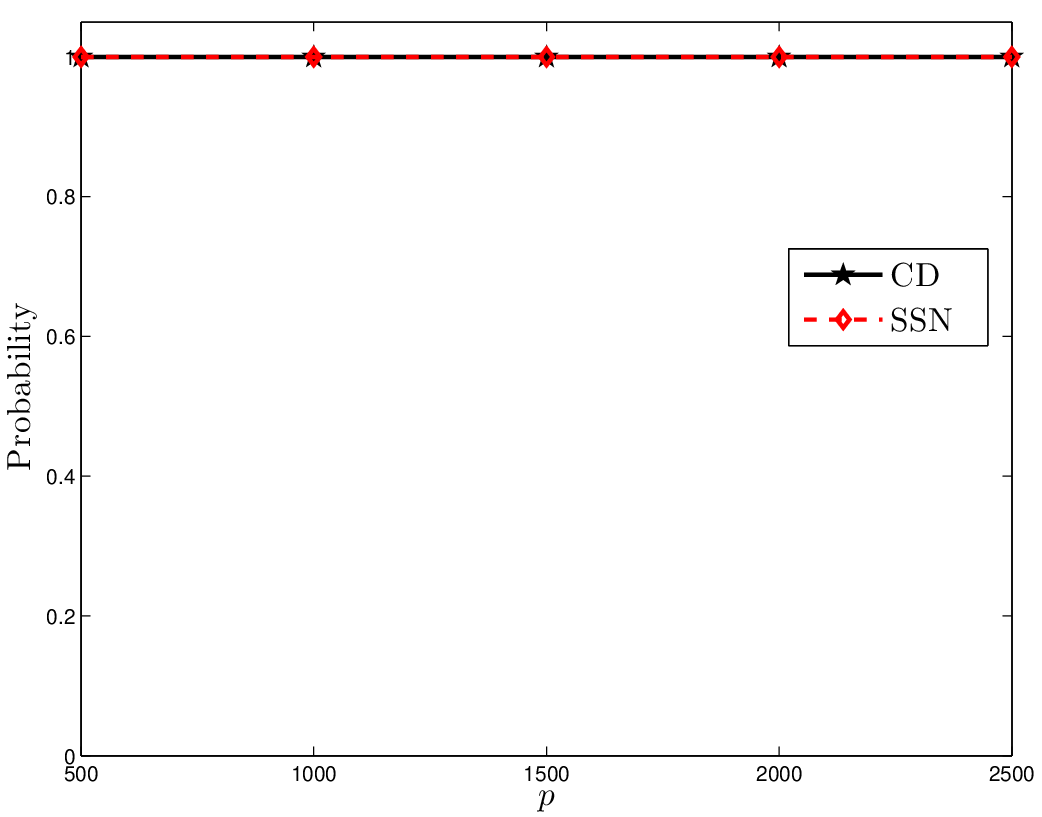}
\end{subfigure}
\begin{subfigure}{.49\linewidth}
\centering
\includegraphics[width=\linewidth]{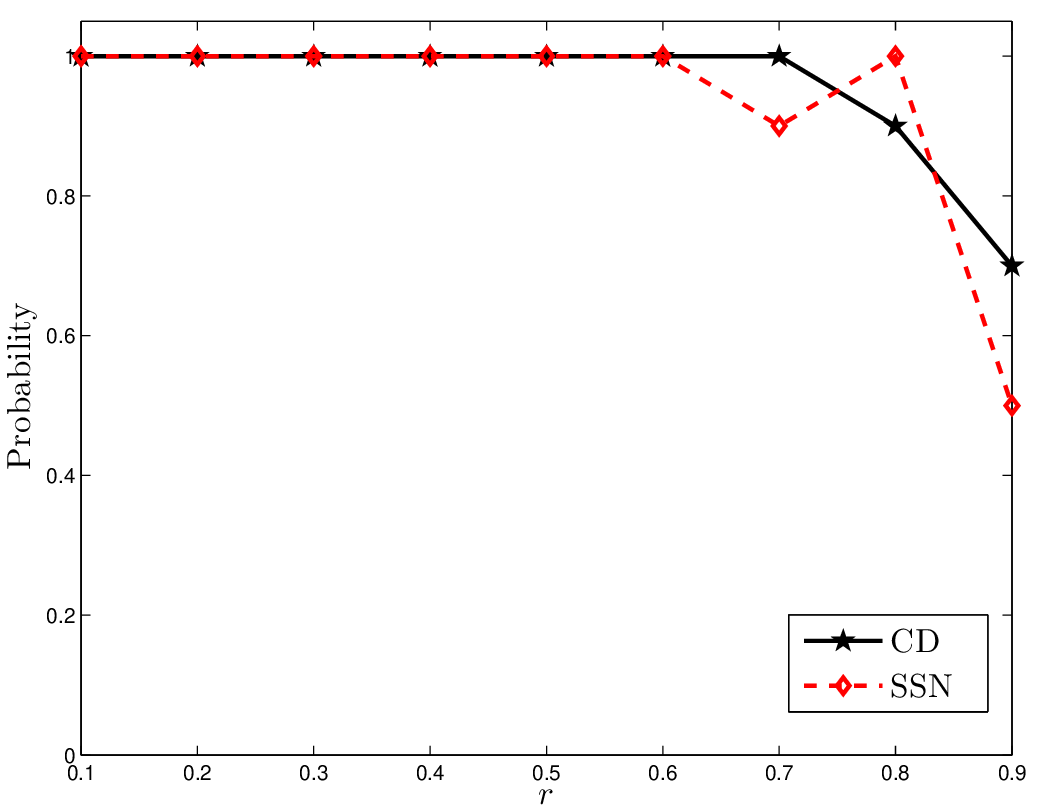}
\end{subfigure}
\begin{subfigure}{.49\linewidth}
\centering
\includegraphics[width=\linewidth]{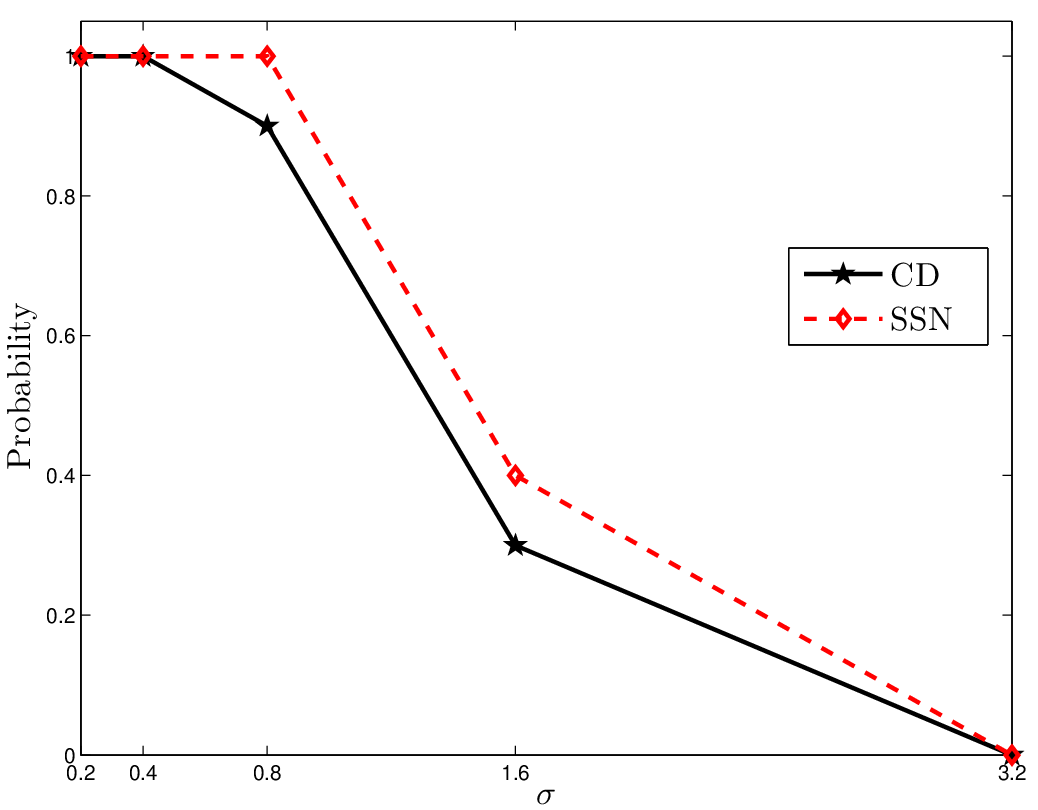}
\end{subfigure}
\begin{subfigure}{.49\linewidth}
\centering
\includegraphics[width=\linewidth]{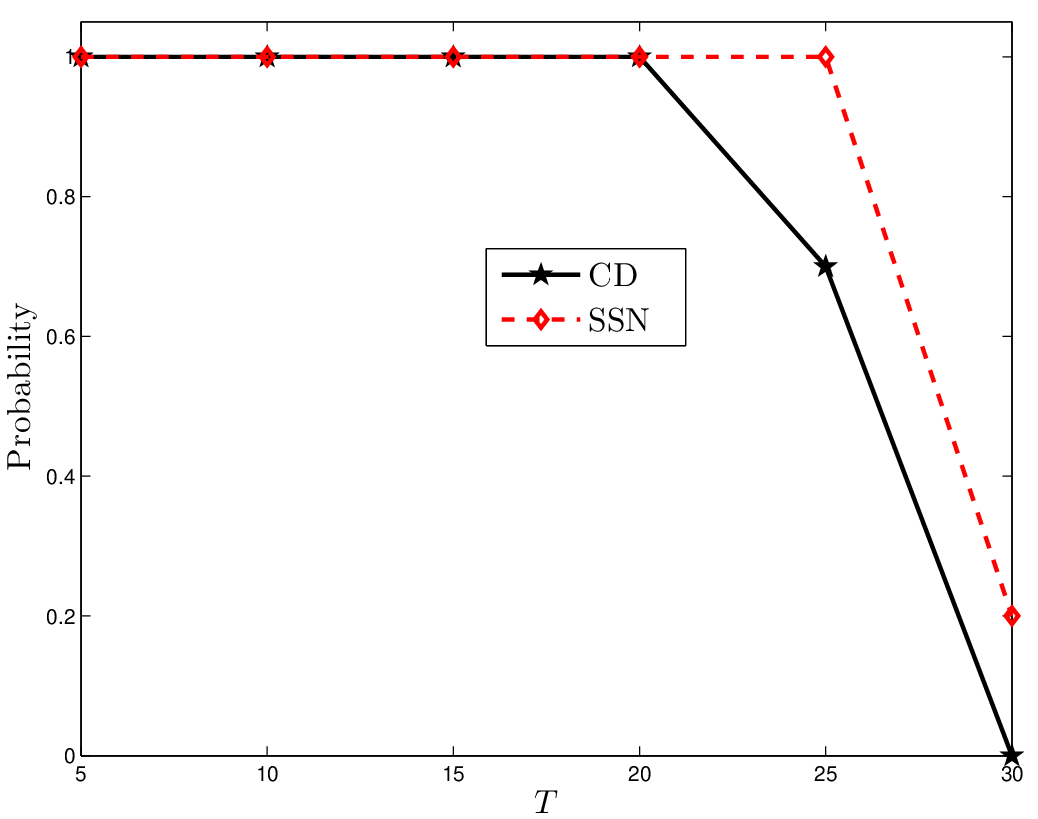}
\end{subfigure}
\caption{
Numerical results of the
influence of the model parameters
on the  exact support recovery probability.
}
\label{fig:IPP}
\end{figure}

\begin{figure}[!ht]
\begin{subfigure}{.49\linewidth}
\centering
\includegraphics[width=\linewidth]{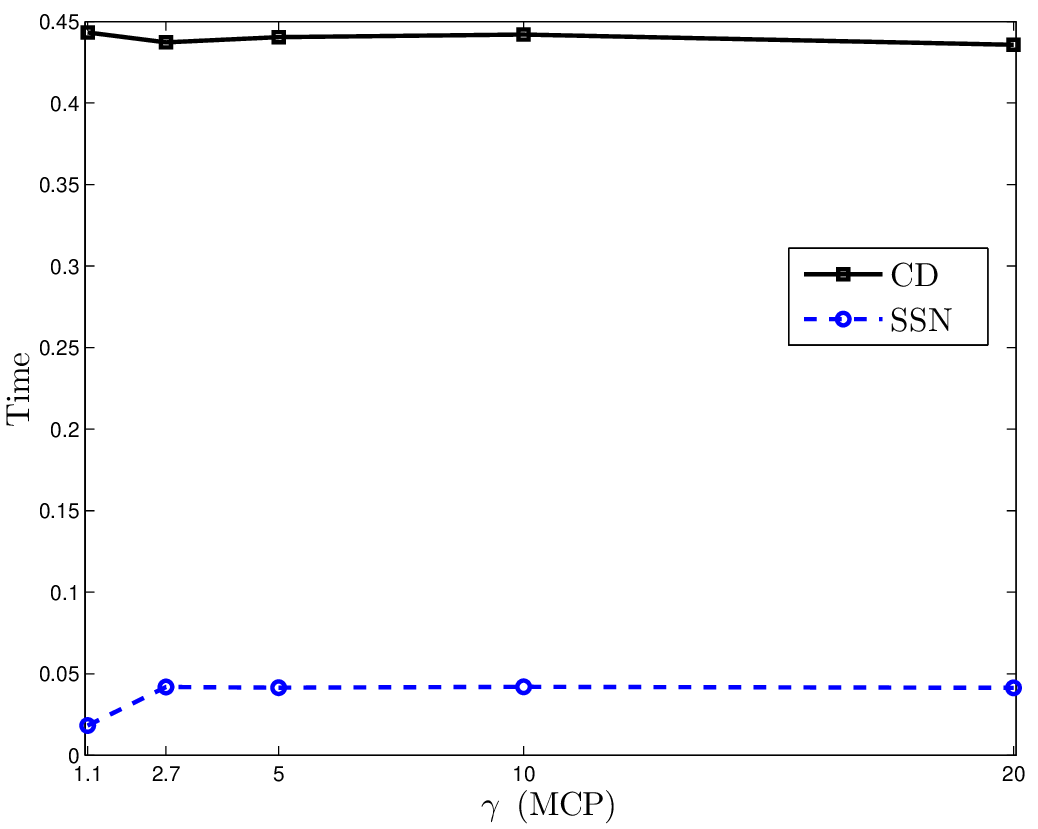}
\end{subfigure}
\begin{subfigure}{.49\linewidth}
\centering
\includegraphics[width=\linewidth]{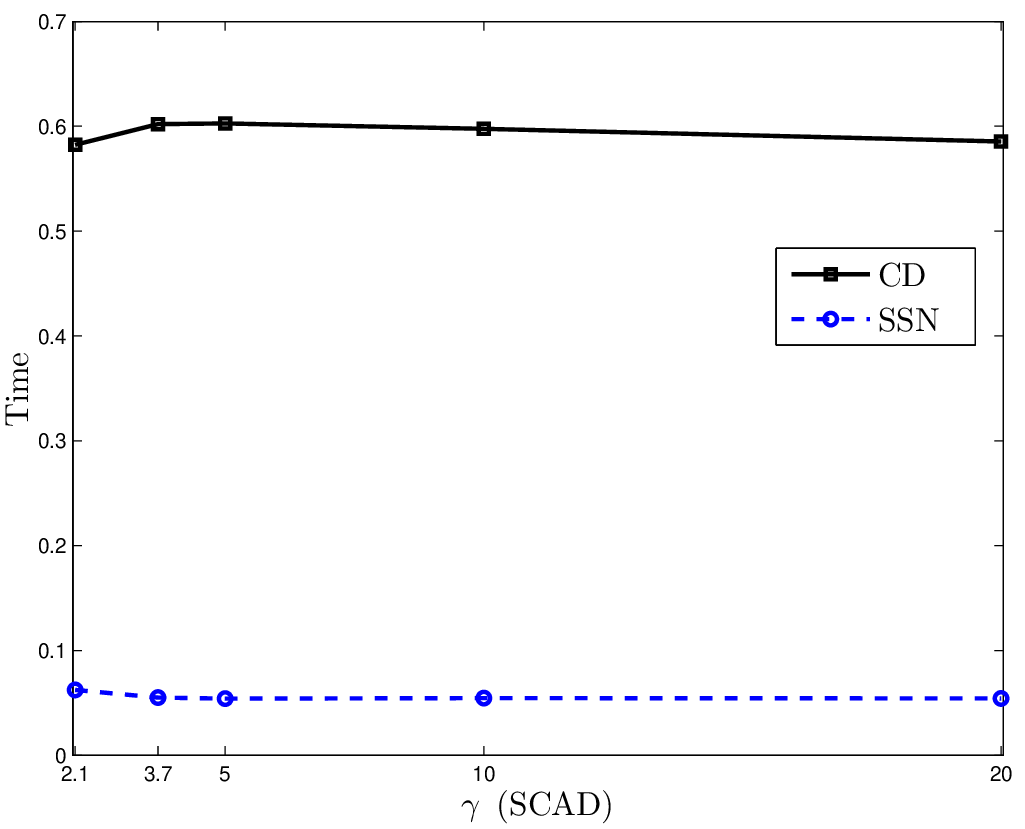}
\end{subfigure}
\begin{subfigure}{.49\linewidth}
\centering
\includegraphics[width=\linewidth]{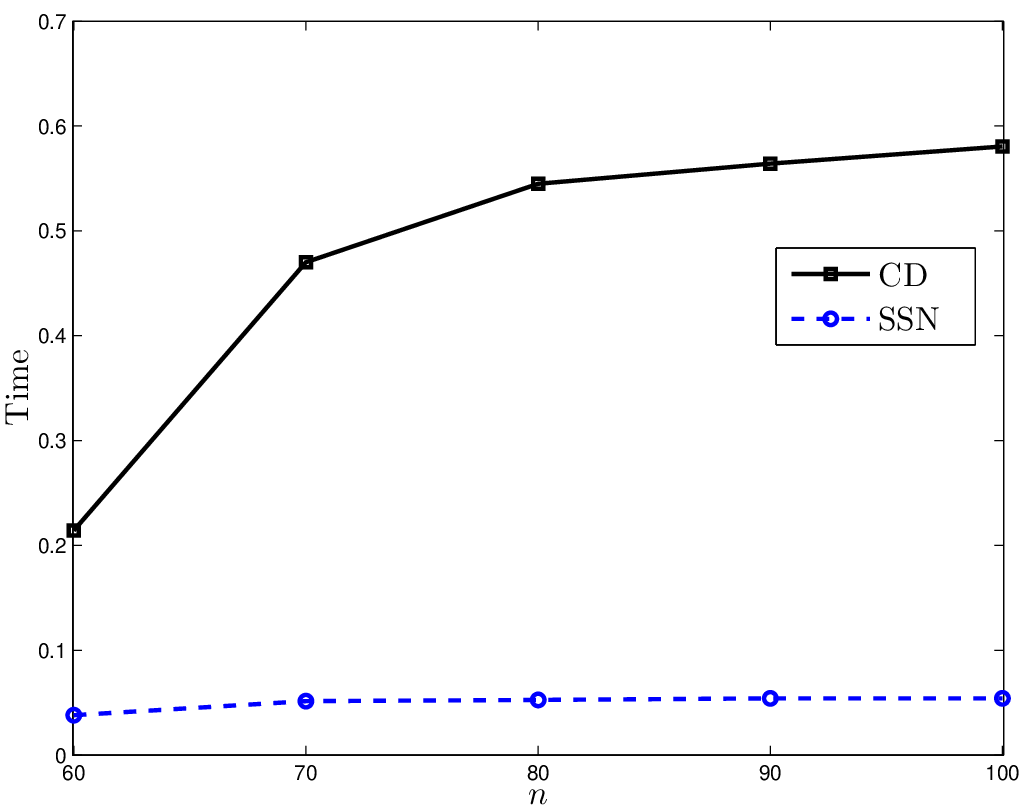}
\end{subfigure}
\begin{subfigure}{.49\linewidth}
\centering
\includegraphics[width=\linewidth]{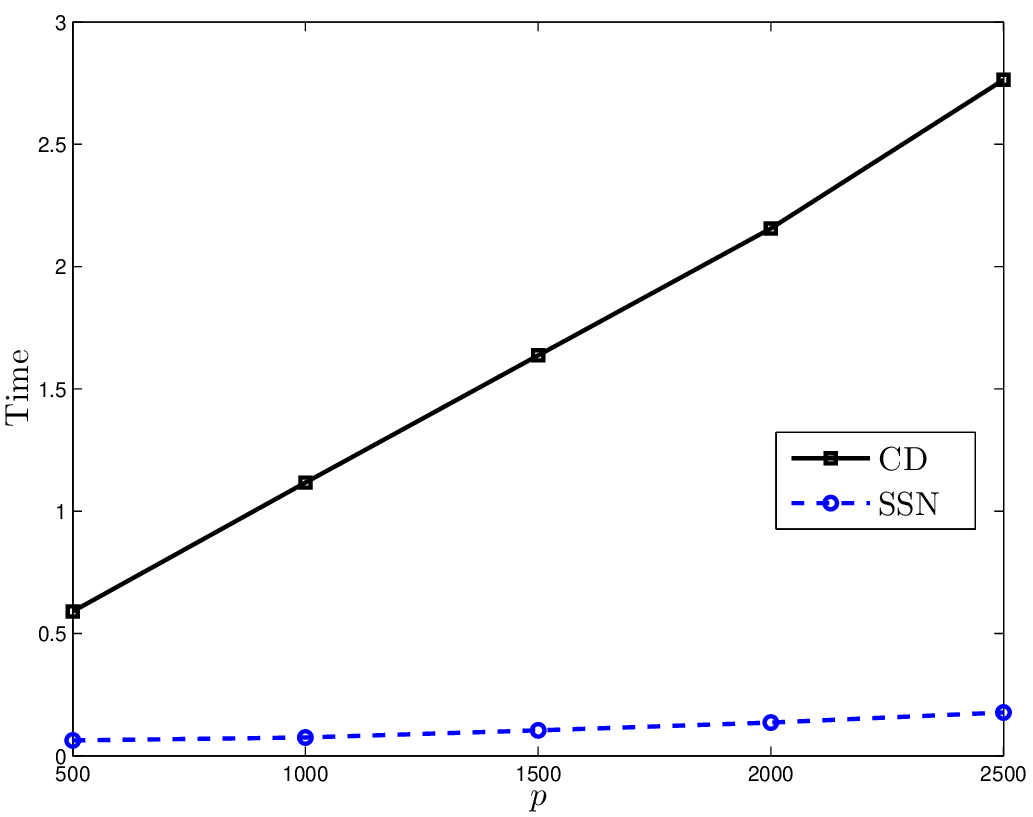}
\end{subfigure}
\begin{subfigure}{.49\linewidth}
\centering
\includegraphics[width=\linewidth]{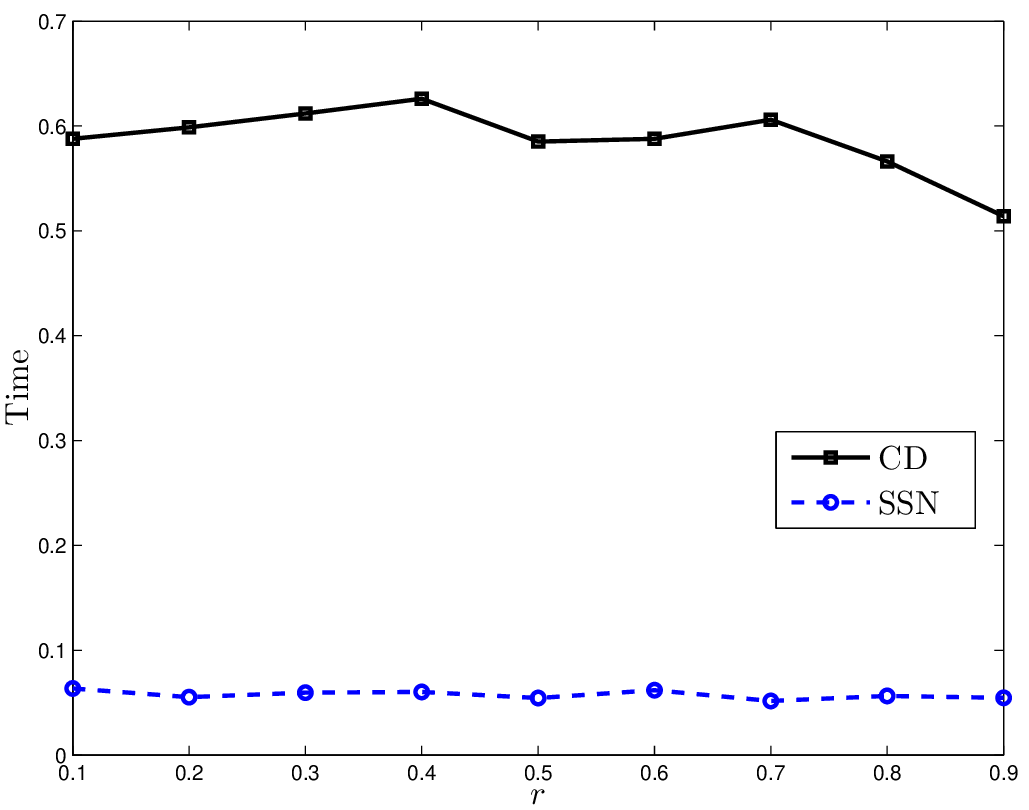}
\end{subfigure}
\begin{subfigure}{.49\linewidth}
\centering
\includegraphics[width=\linewidth]{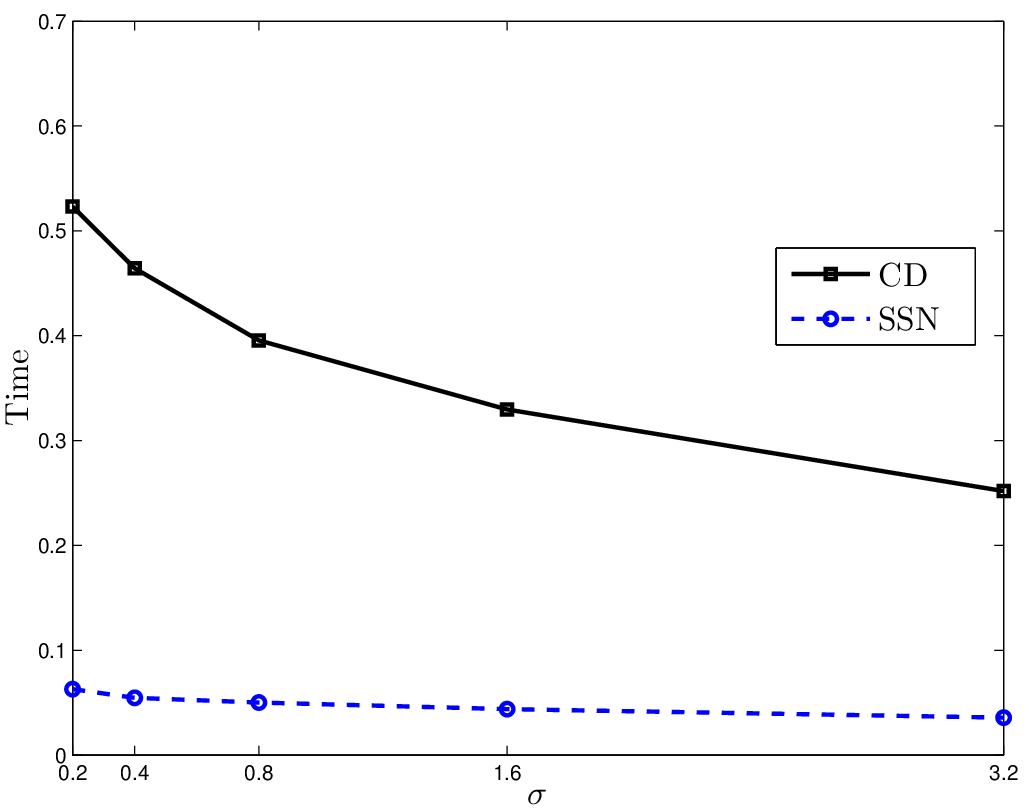}
\end{subfigure}
\begin{subfigure}{.49\linewidth}
\centering
\includegraphics[width=\linewidth]{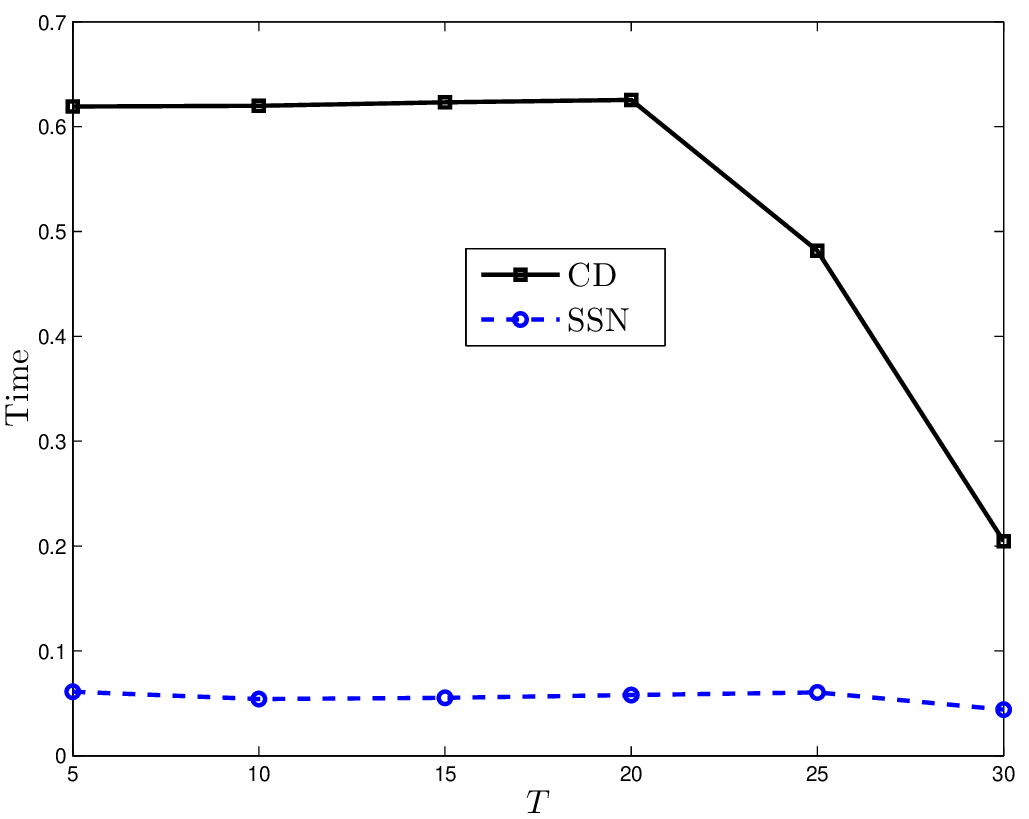}
\end{subfigure}
\caption{
Numerical results of the
influence of the model parameters
on the CPU time.
}
\label{fig:IPT}
\end{figure}

\subsubsection{Breast Cancer Data Example}

Mining biological data are hot issues in machine learning fields
\cite{hastie2009elements,gu2017solution,
mahmud2018applications,cai2018modified,li2019generalized}.
SCAD and MCP are two popular supervised learning methods used
to find sparse structures in high-dimensional data
such as gene expression data \cite{BrehenyP2011,breheny2015group}.
Here we analyze the breast cancer gene expression data
from The Cancer Genome Atlas (TCGA) project
to illustrate the application of SSN in high-dimensional settings.
This data set, which is named bcTCGA,
is available at  \url{http://myweb.uiowa.edu/pbreheny/data/bcTCGA.RData}.
It contains expression measurements of 17814 genes from 536 patients.
There are $491$ genes with missing data, which we have excluded.
We restrict our attention to the $17323$ genes without any missing values.
The response variable $y$ represents the expression levels of gene {BRCA1},
which is the first gene identified that increases the risk of
early onset breast cancer.
The design matrix $X$ is a $536\times 17322$ matrix,
which contains the remaining expression measurements of 17322 genes.
Because {BRCA1} is likely to interact with many other genes,
it is of interest to find genes with expression levels
related to that of {BRCA1}.

We fit a penalized linear regression model to find genes that are related to BRCA1.
We analyze this data set using the proposed SSN algorithm.
We also compare the results with those from the CD algorithm and the adaptive LASSO (ALASSO) \cite{ZouH2006}.
The results from the ALASSO are
computed using the \textbf{ncvreg} \cite{BrehenyP2011}
and \textbf{glmnet} \cite{FriedmanJ2010},
which are R wrappers for C/Fortran.
The following commands complete the main part of the ALASSO computation:
{\footnotesize
\begin{verbatim}
ptm <- proc.time()
library(ncvreg)
fit <- ncvreg(X, y, penalty='lasso')
beta0 <- coef(fit,which=which.min(BIC(fit)))[-1]
weight <- abs(beta0)^(-1)
weight <- pmin(weight,1e10)
library(glmnet); set.seed(0)
cvfit <- cv.glmnet(X,y,nfolds=10,penalty.factor=weight)
betahat <- coef(cvfit, s='lambda.1se')
ptm <- (proc.time()-ptm)[3]
\end{verbatim}
}

We set $\gamma=4$ and $3.7$ for MCP and SCAD, respectively.
\tabref{tab:breast} contains the selected genes
the corresponding estimated coefficients, the CPU time (Time)
and the prediction error (PE).
As shown in this table,
ALASSO, MCP-CD, MCP-SSN, SCAD-CD and SCAD-SSN
identify 8, 5, 10, 13 and 8 genes, respectively.
We observe that some genes, such as
C17orf53, DTL, NBR2, SPAG5 and VPS25,
are selected by at least three solvers, and in particular,
all methods capture the gene NBR2 with relatively large absolute
estimated coefficients. NBR2 is adjacent to
BRCA1 on chromosome 17, and recent experimental evidence indicates that the two
genes share a promoter \cite{breheny2019marginal}.
It is also observed  from \tabref{tab:breast} that
SSN is faster than ALASSO and CD
while producing smaller (or comparable) PE.

\begin{table}[!ht]
\centering
\caption{Analysis of the bcTCGA dataset. Estimated coefficients of different methods are provided. The zero entries correspond to variables omitted.}
\label{tab:breast}
\scalebox{0.78}{
\begin{tabular}{ccccccccc}
\hline
&&&&\multicolumn{2}{c}{MCP}&& \multicolumn{2}{c}{SCAD}\\
\cline{5-6} \cline{8-9}
No.&   Term& Gene&  ALASSO&  CD&  SSN&& CD& SSN \\
\hline
  &       Intercept&         & -0.6786& -0.7361& -0.4888&& -0.9269& -0.4552\\
1 & $\beta_{ 1743}$& C17orf53&  0.1208&  0     &  0.0240&&  0.0979&  0.3883\\
2 & $\beta_{ 2739}$&   CCDC56&  0     &  0     &  0     &&  0.0484&  0     \\
3 & $\beta_{ 2964}$&   CDC25C&  0     &  0     &  0     &&  0.0360&  0.0444\\
4 & $\beta_{ 2987}$&     CDC6&  0     &  0     &  0     &&  0.0089&  0     \\
5 & $\beta_{ 3105}$&    CENPK&  0     &  0     &  0     &&  0.0125&  0     \\
6 & $\beta_{ 4543}$&      DTL&  0.2544&  0.2883&  0.2223&&  0.0886&  0.0746\\
7 & $\beta_{ 6224}$&     GNL1&  0     &  0     & -0.1452&&  0     & -0.1025\\
8 & $\beta_{ 7709}$&  KHDRBS1&  0.1819&  0     &  0     &&  0     &  0     \\
9 & $\beta_{ 7719}$& KIAA0101&  0     &  0     &  0     &&  0     &  0.0675\\
10& $\beta_{ 8782}$&    LSM12&  0     &  0.0432&  0     &&  0     &  0     \\
11& $\beta_{ 9230}$&    MFGE8&  0     &  0     & -0.0273&&  0     & -0.0972\\
12& $\beta_{ 9941}$&     NBR2&  0.4496&  0.3984&  0.5351&&  0.2408&  0.5319\\
13& $\beta_{11091}$&    PCGF1&  0     &  0     & -0.0968&&  0     &  0     \\
14& $\beta_{12146}$&    PSME3&  0.1439&  0     &  0     &&  0.0746&  0     \\
15& $\beta_{13518}$&   SETMAR&  0     &  0     & -0.0807&&  0     &  0     \\
16& $\beta_{14296}$&    SPAG5&  0.0035&  0.0598&  0.2370&&  0.0117&  0.0688\\
17& $\beta_{14397}$&    SPRY2&  0     &  0     &  0     && -0.0058&  0     \\
18& $\beta_{15122}$& TIMELESS&  0     &  0     &  0     &&  0.0322&  0     \\
19& $\beta_{15535}$&    TOP2A&  0     &  0     &  0     &&  0.0299&  0     \\
20& $\beta_{15882}$&   TUBA1B&  0.0857&  0     &  0     &&  0     &  0     \\
21& $\beta_{16315}$&    VPS25&  0.2560&  0.2372&  0.3581&&  0.0998&  0     \\
22& $\beta_{16640}$&   ZBTB26&  0     &  0     & -0.1015&&  0     &  0     \\ \hline
  &            Time&         & 14.7600& 28.7712&  3.3998&& 30.8354&  3.1835\\
  &              PE&         &  0.2116&  0.2344&  0.1837&&  0.2639&  0.2117\\
\hline
\end{tabular}
}
\end{table}

To further evaluate the performance of SSN relative to CD
on the quality of variable selection,
we conduct the cross-validation  procedure similar to
\cite{HuangJ2008,YiC2017}.
We conduct $50$ random partitions of the data.
For each partition,  we randomly choose
$3/4$ observations and $1/4$ observations
as the training and test data, respectively.
We compute the CPU time (Time, in seconds) and the model size (MS, i.e.,
the number of selected genes) using the training data,
and calculate the PE based on the test data.
\tabref{tab:breastCv} presents the average values over $50$ random partitions,
along with corresponding standard deviations in the parentheses.
As shown in \tabref{tab:breastCv},
SSN performs better than CD in terms of the CPU time
while producing comparable PE.
Based on $50$ random partitions,
we also report the selected genes and their corresponding frequency (Freq)
of being selected in \tabref{tab:breastCvFreq}.
We only list genes with frequency greater than or equal to 5 counts.
For each scenario,
the gene NBR2 is the top 1 gene
ranked according to the frequency
and has a significantly higher frequency than other genes,
which is consistent with the findings of \tabref{tab:breast} and \cite{breheny2019marginal},
and largely implies this gene is most closely related to BRCA1.

\begin{table}[!ht]
\caption{The CPU time (Time),  model size (MS) and prediction error (PE)
averaged across $50$ random partitions of the real data
(numbers in parentheses are standard deviations)}
\label{tab:breastCv}
\centering
\scalebox{1}{
\begin{tabular}{ccccc}
\hline
Penalty&  Method&            Time&           MS&             PE\\ \hline
   MCP &      CD& 25.3317(1.4440)& 7.02(2.8820)& 0.2761(0.0451)\\
       &     SSN&  2.5452(0.4913)& 9.46(3.7591)& 0.2914(0.0568)\\
   SCAD&      CD& 26.0625(1.4431)& 9.46(3.3025)& 0.3190(0.0624)\\
       &     SSN&  2.5736(0.4382)& 9.00(4.4584)& 0.2647(0.0471)\\ \hline	
\end{tabular}
}
\end{table}

\begin{table}[!ht]
\caption{Frequency table for $50$ random partitions of the real data.
To save space, only the genes with Freq $\geq 5$ are listed.
}
\label{tab:breastCvFreq}
\begin{center}
\scalebox{0.68}{
\begin{tabular}{lllllllllll} 
\hline
\multicolumn{5}{c}{MCP}&& \multicolumn{5}{c}{SCAD}\\
\cline{1-5} \cline{7-11}
\multicolumn{2}{c}{CD}&&\multicolumn{2}{c}{SSN}&&
\multicolumn{2}{c}{CD}&&\multicolumn{2}{c}{SSN}\\
\cline{1-2} \cline{4-5} \cline{7-8} \cline{10-11}
    Gene& Freq&&Gene     &Freq&&     Gene&Freq&&     Gene&  Freq\\ \hline
NBR2	&50	  &&NBR2	 &46  &&NBR2	 &49  &&NBR2	 &50\\
DTL		&36   &&C17orf53 &40  &&DTL		 &47  &&C17orf53 &34\\
VPS25	&31   &&MFGE8	 &11  &&VPS25	 &47  &&DTL		 &17\\
C17orf53&25   &&VPS25	 &10  &&C17orf53 &47  &&VPS25	 &13\\
CCDC56	&12   &&GNL1	 &10  &&PSME3	 &38  &&MFGE8	 &11\\
CDC25C	&12   &&PSME3	 &8   &&TOP2A	 &34  &&KLHL13	 &10\\
SPAG5	&12   &&C3orf10	 &8   &&CCDC56	 &26  &&C10orf76 &8 \\
PSME3	&12   &&DTL		 &7   &&TIMELESS &26  &&GNL1	 &8 \\
LSM12	&11   &&CMTM5	 &7   &&CDC25C	 &20  &&ZBTB26	 &8 \\
TOP2A	&8    &&FAM103A1 &7   &&CENPK	 &15  &&TIMELESS &7 \\
KLHL13	&8    &&KIAA0101 &7   &&SPAG5	 &14  &&DYNLL1	 &6 \\
SPRY2	&8    &&ZYX		 &6   &&CCDC43	 &12  &&FAM103A1 &6 \\
NPR2	&7    &&MCM6	 &6   &&CDC6	 &12  &&KHDRBS1	 &6 \\
TUBA1B	&7    &&SETMAR	 &6   &&SPRY2	 &9   &&SETMAR	 &6 \\
TIMELESS&6    &&FBXO18	 &6   &&RDM1	 &8   &&TOP2A	 &6 \\
KIAA0101&6    &&SPAG5	 &6   &&CENPQ	 &6   &&CMTM5	 &6 \\
CENPK	&5    &&LMNB1	 &6   &&TUBG1	 &6   &&FGFRL1	 &6 \\
CRBN	&5    &&KHDRBS1	 &6   &&         &    &&PSME3	 &5 \\
VDAC1	&5    &&ATAD1	 &6   &&         &    &&LMNB1	 &5 \\
CMTM5	&5    &&KLHL13	 &6   &&         &    &&ZNF189	 &5 \\
CDC6	&5    &&AASDH	 &5   &&         &    &&CENPQ	 &5 \\
        &     &&YTHDC2	 &5   &&         &    &&HMGN2	 &5 \\
        &     &&DYNLL1	 &5   &&         &    &&CRBN	 &5 \\
        &     &&TUBG1	 &5   &&         &    &&         &  \\
        &     &&HMGN2	 &5   &&         &    &&         &  \\
        &     &&CENPQ	 &5   &&         &    &&         &  \\
        &     &&C10orf30 &5   &&         &    &&         &  \\
\hline
\end{tabular}
}
\end{center}
\end{table}

\subsection{Comparison with the DCPN}
\label{ssec:comDCPN}

We conduct numerical experiments to compare DCPN and SSN.
The  implementation setting and comparison metrics
are given in Sections \ref{ssec:simSet} and \ref{ssec:comCD}, respectively.
The noise level $\sigma=1$ is fixed.
The DCPN is implemented using the picasso package \cite{PCD,picasso-package}.
For both solvers, we use the HBIC selector to choose
the tuning parameters.
We set $\gamma=4$ and $3.7$ for MCP and SCAD, respectively.
Simulation results are summarized in \tabref{tab:simuPicasso}.

\begin{table}[!ht]
\centering
\caption{Simulation results for DCPN and SSN
with $n=\floor{p/5}$ and $T=\floor{n/[2\log(p)]}$
based on 10 independent runs. $(M,J)=(100,1)$ and $\sigma=1$.}
\label{tab:simuPicasso}
\scalebox{0.8}{
\begin{tabular}{ccccccccc}
\hline
 $p$&  $r$&     Penalty& Method&  Time&  MS&  CM&  AE&  RE\\
\hline
 1000&   0.3& MCP  & DCPN& 0.147&  14.0&  100\%&  0.9821&   0.0620\\
	 &      &      & SSN & 0.057&  14.0&  100\%&  0.1456&	0.0142\\
	 &      & SCAD & DCPN& 0.145&  14.0&  100\%&  0.9870&   0.0654\\
	 &      &      & SSN & 0.059&  14.0&  100\%&  0.1456&	0.0142\\
	 &   0.7& MCP  & DCPN& 0.144&  15.0&    0\%&  0.4258&   0.0355\\
	 &      &      & SSN & 0.045&  11.9&   80\%&  1.3921&	0.1255\\
	 &      & SCAD & DCPN& 0.154&  15.0&    0\%&  0.5866&   0.0467\\
	 &      &      & SSN & 0.059&  14.0&  100\%&  0.1693&	0.0149\\
 2000&   0.3& MCP  & DCPN& 0.286&  26.0&  100\%&  1.0080&   0.0623\\
	 &      &      & SSN & 0.209&  26.0&  100\%&  0.1210&	0.0104\\
	 &      & SCAD & DCPN& 0.248&  26.0&  100\%&  1.0436&   0.0715\\
	 &      &      & SSN & 0.222&  26.0&  100\%&  0.1210&	0.0104\\
	 &   0.7& MCP  & DCPN& 0.282&  26.9&   10\%&  0.4845&   0.0424\\
	 &      &      & SSN & 0.200&  23.5&   90\%&  0.9125&	0.0998\\
	 &      & SCAD & DCPN& 0.262&  26.9&   10\%&  0.6178&   0.0546\\
	 &      &      & SSN & 0.222&  26.0&  100\%&  0.1172&	0.0108\\			
\hline
\end{tabular}
}
\end{table}

For each combination, we observed from
\tabref{tab:simuPicasso} that
SSN is faster than DCPN, see the column that includes the values of Time.
Here we should note that SSN is
implemented in Matlab, while DCPN uses the picasso package which is a R wrapped C
solver. So strictly speaking, this is not a fair comparison, but
is actually in favor of DCPN.
In addition, the CPU time generally increases linearly with $p$, and is
relatively robust to the choice of $r$.
When the correlation is low (i.e., $r=0.3$), both solvers
can select the true model 100 percent correctly, 
while SSN can produce solutions with smaller AE and RE comparing with DCPN.
At high correlation level (i.e., $r=0.7$),
the SSN still behaves well in terms of CM, while the DCPN almost fails on
this metric, and both solvers share comparable  AE and RE.
In summary, the comparison results in \tabref{tab:simuPicasso}
show good performance of SSN in terms of both efficiency and accuracy.

\section{Conclusion}
\label{sec:conclusion}

Starting from the KKT conditions
we developed a SSN algorithm for the
MCP and SCAD  regularized learning problems in high-dimensional settings.
We  established the
 local superlinear convergence of SSN and analyzed its
computational complexity.
Combining  with the VC or  HBIC tuning parameter selector
with  warm start,  we obtain the solution path and
select  the tuning parameters in a fast and stable way.
Numerical comparisons with CD and DCPN
demonstrate the efficiency and accuracy of SSN.
A Matlab package \textsl{ssn-nonconvex} that implements the proposed algorithms is available at \url{http://faculty.zuel.edu.cn/tjyjxxy/jyl/list.htm}.

There are several avenues for further research.
First, SSN can be extended to structured sparsity learning problems
such as group MCP/SCAD \cite{BrehenyP2015,jiao2017group},
or to other learning problems with general nonlinear loss functions
\cite{liu2019dualityfree,chang2019accelerated},
especially for those related to deep neural networks (DNNs)
\cite{scardapane2017group,louizos2018learning,ma2019transformed}.
Second, the authors in \cite{DCN}  prove nice statistical convergence
results of their Newton type algorithm in high dimensions.
Inspired by \cite{DCN}, bounding the statistical errors of our SSN algorithm
along the solution path may be
one possible direction to further understanding of the SSN.
Third,  semi-smooth Newton methods  converges to KKT
points supperlinearly but locally,
we adopt simple continuation strategy to globalize it.
Globalization via smoothing  Newton methods
\cite{ChenQiSun:1998,QiSun:1999,QiSunZhou:2000} is  also of immense interest
in future work.


\appendix

Unless otherwise stated, we only
give the proofs for MCP since the counterparts for SCAD are similar.

\subsection{Proof of \lemref{thexp}}
\begin{proof}
See \cite{BrehenyP2011} and \cite{MazumderR2011}.
\end{proof}

\subsection{Proof of \thmref{th1}}

\begin{proof}
Let $\hat{\beta} \in \mathbb{R}^p$ be a global minimizer of
\begin{equation*}
   E_{mcp}(\beta) := \frac{1}{2}\|X\beta
 -y\|^{2}_{2} + \sum_{i=1}^p P_{mcp}(\beta_{i};\lambda, \gamma).
\end{equation*}
Define $$f_{i}(t) = E_{mcp}(\hat{\beta}_{1},...,\hat{\beta}_{i-1},t,\hat{\beta}_{i+1,},...,\hat{\beta}_{p}),i=1,2,...,p.$$ Then, by definition
$\hat{\beta}_{i}$ is a global minimizer of the scalar function $f_{i}(t)$.
Some algebra shows that the minimizer of $f_{i}(t)$ is also the minimizer of $\frac{1}{2}(t-(\hat{\beta}_{i}+X_{i}^{T}(y -X\hat{\beta})))^2 + P_{mcp}(t;\lambda,\gamma)$.
Then  $\hat{\beta}_{i}=T_{mcp}(\hat{\beta}_{i}+X_{i}^{T}(y -X\hat{\beta});\lambda,\gamma)$ by \lemref{thexp}.
\eqref{K1} - \eqref{K2} follow from the definitions of
$G = X^{T}X,$  $\tilde{y} = X^{T}y$ and $\mathbb{T}(\cdot;\lambda, \gamma)$.

Conversely, assume that  $(\hat{\beta}, \hat{d})$
satisfies \eqref{K1} - \eqref{K2}. Then $\forall i,  \hat{\beta}_{i}=T_{mcp}(\hat{\beta}_{i}+X_{i}^{T}(y -X\hat{\beta});\lambda,\gamma)$, which implies  that
 $\hat{\beta}_{i}$ is a minimizer of the $f_{i}(t)$ defined above. This shows  $\hat{\beta}$ is a coordinate-wise minimizer of $E_{mcp}$.
 By Lemma 3.1 in \cite{TsengP2001}, $\hat{\beta}$ is also a stationary point.
\end{proof}

\subsection{Proof of \lemref{egN}}

\begin{proof}
Observing that both
$T_{scad}(\cdot;\lambda,\gamma)$ and $T_{mcp}(\cdot;\lambda,\gamma)$ are piecewise linear functions on $\mathbb{R}^{1}$, it suffices to prove that the  piecewise linear scalar function
\begin{equation*}
f(t)=
\left\{
\begin{aligned}
   &k_{1}t+b_1, \quad if \quad t \leq t_{0},\\
   &k_{2}t+b_2, \quad   if \quad t >t_{0}.
\end{aligned}
\right.
\end{equation*}
is Newton differentiable with
\begin{equation}\label{ndexa}
 \nabla_{N} f(t) =
  \left\{
    \begin{array}{ll}
   $\text{$k_{1}$}$,    \quad &\text{$ t<t_{0},$}\\
   $\text{$ r $ }$ \in \mathbb{R}^{1},  \quad &\text{$t  = t_{0},$}\\
   $\text{$k_2$}$  ,  \quad &\text{$t_{0}<t.$}\\
    \end{array}
  \right.
\end{equation}
where $k_{1}, k_{2},b_{1}, b_{2},t_{0}$ are
any constants satisfying $k_{1}t_{0}+b_{1}= k_{2}t_{0}+b_{2}$.
Indeed, let
\begin{equation*}
 D(t+h)h =
  \left\{
    \begin{array}{ll}
   $\text{$k_{1}h$}$,    \quad &\text{$ t<t_{0},$}\\
   $\text{$k_{2}h$}$ ,  \quad &\text{$t_{0}<t,$}\\
    \end{array}
  \right.
\end{equation*}
   and $G(t_{0})h = rh$ with an arbitrary $r\in \mathbb{R}$,
  then $ |f(t+h)-f(t)-D(t+h)h|/|h| \longrightarrow0$ as $|h| \longrightarrow0.$
\end{proof}

\subsection{Proof of \thmref{th4}}

\begin{proof}
Let
\begin{equation*}
 b(t;\lambda,\gamma) =
  \left\{
    \begin{array}{ll}
   $ 0$,    \quad &\text{$|t|\leq\lambda,$}\\
   $ 1/(1-1/$\gamma$)$  ,  \quad &\text{$\lambda<|t|<\gamma\lambda,$}\\
   $ 1$  ,  \quad &\text{$\gamma\lambda\geq|t|,$}\\
    \end{array}
  \right.
\end{equation*}
 $\textbf{b}(x;\lambda,\gamma) = \textrm{diag}(b(x_1;\lambda,\gamma),...,b(x_p;\lambda,\gamma))$,
 $g_{i}(x)= T_{mcp}(e_{i}^{T}x;\lambda,\gamma): x\in
\mathbb{R}^{p}\rightarrow
\mathbb{R}^{1}, i= 1,\ldots, p$,
and $\mathbb{T}(x;\lambda, \gamma) = (g_{1}(x),...,g_{p}(x))^T$,
where the column vector $e_{i}$ is the $i_{th}$
orthonormal  basis in $\mathbb{R}^{p}$.

It follows from \lemref{egN} and (\ref{nd2})-(\ref{nd4})
that $b(t;\lambda,\gamma)\in \nabla_{N} {T}_{mcp}(t)$ and
\begin{equation}\label{inter1}
\textbf{b}(x;\lambda,\gamma)\in \nabla_{N} \mathbb{T}(x;\lambda, \gamma).
\end{equation}
Then, by \eqref{inter1} and (\ref{nd2})-(\ref{nd4}) the vector value function $ F_{1}(\beta;d)$ is
 Newton differentiable and
 \begin{equation}\label{n1}
\begin{bmatrix}
H_{11}^{k}       &H_{12}^{k} \\
\end{bmatrix}
\in \nabla_{N} F_{1}(z^{k}_{mcp}),
\end{equation}
where $H_{11}^{k}$ and  $H_{12}^{k}$ are given in \eqref{aicmcp9}.
  By (\ref{nd2})-(\ref{nd4}), $ F_{2}(\beta;d) $ is  Newton differentiable and
   \begin{equation}\label{n2}
   \begin{bmatrix}
H_{21}^{k}       &H_{22}^{k} \\
\end{bmatrix}
\in \nabla_{N} F_{2}(z^{k}_{mcp}),
\end{equation}
where $H_{21}^{k}$ and  $H_{22}^{k}$ are also given in \eqref{aicmcp9}.
It follows from
\eqref{n1}-\eqref{n2} and  (\ref{nd2})-(\ref{nd4}) that
$H_{mcp}^{k}
\in \nabla_{N} F_{mcp}(z^{k}_{mcp})$.

The uniform boundedness of $(H_{mcp}^{k})^{-1} $  is derived similarly as the proof of Theorem 2.6 in \cite{YiC2017}.
\end{proof}

\subsection{Proof of \thmref{th5}}

\begin{proof}
Let  $z^{k}_{mcp}$ be sufficiently close to  $z^{*}_{mcp}$,
which is  a root of $F_{mcp}$.  By the definition of the Newton derivative, we have
\begin{equation}\label{pth5}
\begin{aligned}
&\|H^k_{mcp}(z^{k}_{mcp}-z^{*}_{mcp})-F_{mcp}(z^{k}_{mcp}) + F_{mcp}(z^{*}_{mcp})\|_{2}\\
&\leq  \epsilon \|z^{k}_{mcp}-z^{*}_{mcp}\|_{2},
\end{aligned}
\end{equation}
 where $\epsilon\rightarrow 0$ as  $z^{k}_{mcp}\rightarrow z^{*}_{mcp}$.
 Then, by the definition of SSN  and the fact that $F_{mcp}(z^{*}_{mcp})=0$, we get
 {\footnotesize
\begin{align*}
&\|z^{k+1}_{mcp}-z^{*}_{mcp}\|_{2} \\
&= \|z^{k}_{mcp}-(H^k_{mcp})^{-1}F_{mcp}(z^{k}_{mcp})-z^{*}_{mcp}\|_{2} \\
&=\|z^{k}_{mcp}-(H^k_{mcp})^{-1}F_{mcp}(z^{k}_{mcp})-z^{*}_{mcp} + (H^k_{mcp})^{-1}F_{mcp}(z^{*}_{mcp})\|_{2}\\
&\leq \|(H^k_{mcp})^{-1}\| \|H^k_{mcp}(z^{k}_{mcp}-z^{*}_{mcp})-F_{mcp}(z^{k}_{mcp}) + F_{mcp}(z^{*}_{mcp})\|_{2}. \\
&\leq M_{\gamma} \epsilon \|z^{k}_{mcp}-z^{*}_{mcp}\|_{2}.
\end{align*}
}
The last inequality follows from (\ref{pth5})
and  the uniform boundedness of $(H^k_{mcp})^{-1}$ proved in \thmref{th4}.
Therefore, the sequence  $z_{mcp}^{k}$ generated by \algref{algmcp}
converges to $z_{mcp}^{*}$ locally and superlinearly.
\end{proof}

\section*{Acknowledgment}

The authors are grateful to the editor, the associate editor
and the referees for their many constructive and
insightful comments that led to significant improvements in the
article. The authors also would like to thank Professor Defeng
Sun for helpful discussions on the semi-smooth Newton method
and related topics, and appreciate
Doctor Yicheng Kang of Bentley University for his careful reading
of the paper and helpful suggestions in the writing.

\bibliographystyle{IEEEtran}
\bibliography{IEEEabrv,ref_ssn_v10}





\end{document}